\documentclass[]{interact}
\usepackage{xcolor}
\usepackage{epstopdf}% To incorporate .eps illustrations using PDFLaTeX, etc.
\usepackage[caption=false]{subfig}% Support for small, `sub' figures and tables

\usepackage[longnamesfirst,sort]{natbib}% Citation support using natbib.sty
\bibpunct[, ]{(}{)}{;}{a}{,}{,}% Citation support using natbib.sty
\renewcommand\bibfont{\fontsize{10}{12}\selectfont}% To set the list of references in 10 point font using natbib.sty

\theoremstyle{plain}% Theorem-like structures provided by amsthm.sty
\newtheorem{theorem}{Theorem}[section]
\newtheorem{lemma}[theorem]{Lemma}

\newtheorem{proposition}[theorem]{Proposition}

\theoremstyle{definition}

\theoremstyle{remark}
\newtheorem{remark}{Remark}

\begin{document}

\articletype{Research articles}% Specify the article type or omit as appropriate

\title{Parameter estimation {for} fractional autoregressive process with seasonal structure }

\author{
\name{Chunhao Cai\textsuperscript{a}\thanks{CONTACT Yiwu Shang. Email: shangyiwu@mail.nankai.edu.cn} and Yiwu Shang\textsuperscript{b}}
\affil{\textsuperscript{a}School of mathematics(Zhuhai),Sun Yat-sen University,Zhuhai,People’s Republic of China; \textsuperscript{b}School of mathematical sciences,Nankai University,Tianjin,People’s Republic of China}
}

\maketitle
 
\begin{abstract}
This paper introduces a new kind of seasonal fractional autoregressive process (SFAR) driven by fractional Gaussian noise (fGn). The new model includes a standard seasonal AR model and fGn. {The estimation of the parameters of this new model has to solve two problems: nonstationarity from the seasonal structure and long memory from fGn. We innovatively solve these by getting a stationary subsequence, making a stationary additive sequence, and then obtaining their spectral density. Then, we use one-step procedure for Generalized Least Squares Estimator (GLSE) and the Geweke Porter-Hudak (GPH) method to get better results.  We prove that both the initial and one-step estimators are consistent and asymptotically normal. Finally, we use Monte Carlo simulations with finite-sized samples to demonstrate the performance of these estimators. Moreover, through empirical analysis, it is shown that the SFAR model can simulate some real world phenomena better than general models.}   
\end{abstract}

\section{INTRODUCTION}  %--> sections in UPPERCASE

The long memory phenomenon and seasonal phenomenon play important roles in economics, geography, and other fields. One classic type of seasonal model with long memory is the  $ARFISMA{(p,d,q)\times(P,D,Q)_{s}}$
  process, which has been extensively researched by  \cite{hosking1984modeling}, \cite{franco2007bootstrap},  \cite{chan1995inference}.  As demonstrated below, these models are equivalent to the ARUMA model, which can be expressed as
\begin{equation}
{G(\mathcal{B})}X_{n} = \epsilon_{n}, \label{1.4}
\end{equation}
	 where $\mathcal{B}$ is the lag operator, $\epsilon_{n}$ is a short memory process, { and $G(\cdot)$ satisfy the recurrence relation} 
\begin{equation}
	{G(z)} =  (1-z)^{d_{0}}\left\{ \prod_{j=1}^{r-1}(1-2zcos\lambda_{j}+z^{2})^{d_{j}} \right\}(1+z)^{d_{r}} , \label{1.5}
\end{equation}
	{
where \(|d_{j}|\leq\frac{1}{2}\), and for \(j = 0,1,2,\dots,r\), denoting that \(r\) is a positive integer, the frequencies satisfy \(0\leq\lambda_{j}\leq\pi\). Additionally, \(\lambda_{0} = 0\) and \(\lambda_{r}=\pi\).}
\par {In the present paper, we conduct a comprehensive study on a \(p\)-order seasonal fractional autoregressive process (SFAR), denoted as \(X_{nT + u}\).
Here, the nonnegative integer \(T\) signifies the number of seasons, suggesting that the time series data exhibits seasonal fluctuations with a period of \(T\). For example, when \(T = 4\), it represents a quarterly seasonal pattern, and when \(T = 12\), it corresponds to a monthly seasonal pattern.
For any \(n \in\mathbb{N}\) (where $n$ denotes the count of complete seasonal cycles), the model adheres to the following recursive relation:}
 \begin{equation}
	X_{nT+u} =\sum_{i=1}^{p}\phi_{i}(nT+u)X_{(n-i)T+u}+\epsilon^{H}_{nT+u}, \quad u= 1,2,...T,  \label{1.1}  
 \end{equation}
{where $u$ denotes the specific time points within each seasonal cycle, thus taking values from $1$ to $T$, $p$ represents the order of the autoregressive part of the model. The $\phi_{i}(nT + u)$ are autoregressive seasonal coefficients, which may change with time and satisfy $\phi_{i}(u)=\phi_{i}(u + T)$. $\epsilon_{nT + u}^{H}$ represents fractional Gaussian noise, which explains the nonseasonal fluctuations.}
	 Fractional Gaussian noise exhabits long memory when $\frac{1}{2}<H<1 $. The long memory phenomenon indicates strong autocorrelation or dependence in time series data. We typically say that $X_{t}$ has long memory if its covariance satisfies 
\begin{equation}
 \gamma_{j} \sim {Q}j^{2H-2},\quad j \rightarrow \infty, \label{1.2}
 \end{equation}
	the spectral density is defined by the scheme 
\begin{equation}
 f(\lambda) \sim {V}\lambda^{1-2H},\quad \lambda \rightarrow 0^{+}, \label{1.3}
\end{equation}
	where $\frac{1}{2} < H < 1$, {Q} and {V} are constants greater than 0. \cite{robinson2010long}, \cite{bisognin2009properties}, \cite{beran2013long}  did a great deal of detailed and excellent work in fractional Gaussian noise (fGn), especially in the estimation of \(H\). 
   \par  {The Seasonal Fractional Autoregressive (SFAR) model represents a natural expansion of the fractional autoregressive process (FAR). The FAR, recognized as a long memory model is formulated as 
\begin{equation}
X_{n}=\sum_{i = 1}^{p}a_{i}X_{n - 1 - i}+\epsilon_{n}^{H}, \quad  n\in\mathbb{N}.
\end{equation}
where $a_{i} \in \mathbb{R}$. It is composed of fractional Gaussian noise, and its long range dependence characteristics are determined by the value of $H$.} 
\par {In this paper, we focus on the estimation and asymptotic properties of parameters in the SFAR model. \cite {geweke1983estimation} and \cite{carlin1989sensitivity} conducted some research on such models in the early stage. We aim to extend related research and will use a one-step procedure to optimize our approach.} 
 \par {For the parameter estimation of SFAR model, two key problems need to be addressed: the nonstationarity resulting from the seasonal structure and the dependence within the fractional Gaussian noise.}

\par Seasonality is a distinctive feature of time series data where patterns repeat at regular intervals, typically defined by a specific period T. Seasonal time series models are often nonstationary, which presents certain challenges for our research. A common solution is to perform seasonal differencing on the time series. Seasonal autoregressive process is a classical model was proposed by \cite{harrison1965short} and \cite{chatfield1973box}. \cite{tsay2013multivariate} has explored SAR model with white noise  in detail, however, there still remain many interesting variations worthy of research. For instance, the study in \cite{kong2023seasonal} employed particle filtering likelihood methods to estimate seasonal count time series. The other category is the research on the SAR model driven by fGn.
\par { Previous studies by \citet{brouste2014asymptotic} and \citet{soltane2024asymptotic} have laid a foundation for the estimation of the parameters $\phi_{i}(u)$ in FAR models.
In this paper, we use the modified Generalized Least Squares Estimation (GLSE) proposed by \citet{esstafa2019long} and \citet{hariz2024fast} to obtain a consistent estimator of $\phi_{i}(u)$. Additionally, we will prove that this estimator is asymptotically normal.}
 \par {Time series models with long memory show long range dependencies between distant observations, posing challenges to traditional statistical analysis and forecasting. In the SFAR model, long memory comes from fractional Gaussian noise, where the parameter $H$ determines this characteristic. Thus, estimating $H$ is crucial. The first method for estimating $H$ was the rescaled range analysis by \citet{hurst1951long}, but its lack of a limiting distribution complicates statistical inference. Now, popular estimation techniques are the GPH estimation by \citet{geweke1983estimation} and the local Whittle estimation by \citet{robinson1995log}.  }
 \par {  For the estimation of the Hurst index $H$, we will adopt the Geweke Porter-Hudak (GPH) method, which exhibits a smaller bias, for an additive stationary time series derived from the samples. It is worth noting that it would be more straightforward to estimate \(\hat{H}_{n}(u)\) by \((X_{nT + u})_{n\in\mathbb{N}}\). However, this approach is not fundamentally different from the method in \citet{hariz2024fast} and each  \(\hat{H}_{n}(u)\) cannot contain information about all the data. Meanwhile, considering that sequence \((X_{nT + u})_{n\in\mathbb{N},u=1,\dots,T}\) represents data of the same nature, we assume that the long memory parameter is the same for each season and is independent of the season $u$, and the differences between different seasons are only determined by the seasonal parameters. To obtain a unique \(\hat{H}_{n}\), we sum up the data in each cycle to obtain a new sequence $(Y_{n})_{n \in \mathbb{N}}$, then we prove the stationarity of $(Y_{n})_{n \in \mathbb{N}}$, calculate its spectral density, and finally use the GPH method to get  \(\hat{H}_{n}(u)\). This improvement enables us to address the issue of parameter estimation for \(H\) in nonstationary time series with seasonality. }
\par After obtaining the initial estimators of $\phi_{i}(u)$ and  $H$, we modify our approach using a faster and asymptotically efficient method known as the one-step estimator. This method, first proposed by \cite{le1956asymptotic}, has been widely applied in ergodic Markov chains \citep{kutoyants2016multi}, diffusion processes \citep{gloter2021adaptive}, and fractional autoregressive processes \citep{hariz2024fast}. 
The primary challenge lies in calculating the Fisher information matrix, as discussed in \cite{cohen:hal-00638121}. To tackle this issue, We extract the data from each season to form a new series, proving the stationarity of this new series and deriving its spectral density. Subsequently, we can utilize the results from \cite{cohen:hal-00638121} and \cite{hariz2024fast} to obtain related findings.	
	\par 
     This paper is organized as follows. Sections 2 and 3 present the main results. Section 2 introduces the initial estimator of the Hurst index, $\phi_{i}(u)$ and discusses its asymptotic properties. Section 3 derives the one-step estimator and its asymptotic properties. Section 4 provides numerical illustrations to demonstrate the performance of both the initial and one-step estimators. Section 5 concludes the paper and considers the prospects and significance of our research. { Section 6 illustrates that the SFAR model is superior to the traditional seasonal autoregressive model through a practical application. All technical proofs are gathered in Section 7, while Section 8 presents auxiliary results.}

\section{Initital estimator of SFAR(1) models}
\label{sec2}
\subsection{Problem statements and assumptions}
  \par {Without loss of generality, based on the representation of the SFAR model in (\ref{1.1}), we can consider the first order model in this paper and denote $\phi_{1}(u)=\phi(u)$.
  }
  \par {$X_{nT+u}$ is said to be a SFAR(1) model if it admits the representation}
  \begin{equation}
  	X_{nT+u}=\phi(nT+u)X_{(n-1)T+u}+\epsilon^{H}_{nT+u}, \quad u=1,2,3...T, \quad n \in \mathbb{N},   \label{2.1}
  \end{equation}
   where $\phi(u) = \phi(u+nT)$, $T$ represents the season length and $u$ denotes the $u$-th season of the $n$-th cycle. {The term $\epsilon_{nT + u}^{H}$ represents a stationary fractional Gaussian noise with a Hurst index $H$. It is defined as the increment of the fractional Brownian motion, specifically $\epsilon_{nT + u}^{H}=B_{nT + u + 1}^{H}-B_{nT + u}^{H}$, where $B_{nT + u}^{H}$ is the fractional Brownian motion.} The autocovariance of sequence  $(\epsilon^{H}_{n})_{n\in \mathbb{N}}$ takes the form of 
  \begin{equation}
  	\rho(k) = \frac{1}{2}(|k+1|^{2H}-2|k|^{2H}+|k-1|^{2H}), \label{2.2}
  \end{equation}
  the spectral density of  $(\epsilon^{H}_{n})_{n\in \mathbb{N}}$ defined by
  \begin{equation}
  	{f_{\epsilon_{n}^{H}}(\lambda)} = C_{H}(1-cos(\lambda))\sum\limits_{j \in Z}\frac{1}{|\lambda+2j\pi|^{2H+1}}, \label{2.3}
  \end{equation}
where $C_H=\frac{1}{2\pi}\Gamma(2H+1)sin(\pi H)$ and $\lambda \in [-\pi,\pi]$, $\Gamma(\cdot)$ is Gamma function. 
  \par Here are some assumptions and notations bellow.
  \par {${A_{0}}$: Denote \(\Theta_{u}^{l^\star}\) as a compact set with the following expression,
\begin{center}
\(\Theta_{u}^{l^\star}=\left\{\phi(u)\in\mathbb{R}; \text{ the roots of }1 - \phi(u)z = 0\text{ have modulus }\geq1 + l\right\}\).
\end{center}}
  \par {We define the set  $\Theta^{l}_{u}$ as the Cartesian product  $\Theta^{l^\star}_{u} \times [d_{1},d_{2}]$, where $l$ is a positive constant and $[d_{1},d_{2}] \in (0,1)$.}
  \par {${A_{1}}$: \(\phi(u) \in (-1, 1)\) and \(H \in (0, 1)\). }
 \par {\textbf{Notation:} By $\xrightarrow{\mathcal{L}}$ and $\xrightarrow{\mathbb{P}}$, respectively, we denote convergence in law and convergence in probability. Let $\underline{\phi}=(\phi(1),\phi(2),\ldots,\phi(T))$. Denote the parameters $\theta(u)=(\phi(u),H)$, where $\theta(u)\in\mathring{\Theta}^{l}_{u}$, and $\mathring{\Theta}^{l}_{u}$ represents the interior of $\Theta^{l}_{u}$.}
 \par {Define the parameter space $\Theta^{l}=\Theta^{l^{\star}}_{1}\times\Theta^{l^{\star}}_{2}\times\cdots\times\Theta^{l^{\star}}_{T}\times[d_{1},d_{2}]$, which encompasses all the required parameters. Given samples of size $n$, we obtain the estimators $\hat{\theta}_{n}=(\hat{\phi}_{n}(1),\hat{\phi}_{n}(2),\ldots,\hat{\phi}_{n}(T),\hat{H}_{n})$ and $\hat{\theta}_{n}(u)=(\hat{\phi}_{n}(u),\hat{H}_{n})$. 
}
   
  \par  In this paper, we will present both the initial estimator and the one-step estimator for the parameters of the SFAR(1) model. The following sections will delve into the asymptotic properties and characteristics of these estimators in detail.  
  \subsection{The GPH estimator for the hurst index}
  Due to the nonstationarity of $X_{n}$, obtaining an estimator for $H$ using standard semiparametric methods is not feasible. To address this, we can extract stationarity from the data by splitting the time series $(X_{n})_{n \in \mathbb{N}}$ into seasonal components, resulting in $T$ stationary subsequences $\underline{ {X}}(u)=(X_{u},X_{T+u},...,X_{nT+u}) $ and we construct a stationary additive series defined as $Y_{n} = \sum^{T}_{u=1} X_{nT+u}$. 
  \par In this subsection, we will estimate $H$  using the log-periodogram method, specifically the GPH estimator, applied to the additive series$(Y_{n})_{n \in \mathbb{N}}$. The spectral density and stationarity properties of $(Y_{n})_{n \in \mathbb{N}}$ and $(X_{nT+u})_{u\in \mathbb{Z}}$  are outlined in the following three propositions.
  \begin{proposition}
  	For each $u = 1,2,...T$ and any $n\in \mathbb{N}$, Under conditions $(A_{0})$ and $(A_{1})$, the process
  	\begin{equation}
  		X_{nT+u} = \sum^{\infty}_{j=0} \phi^{j}(u)\epsilon^{H}_{(n-j)T+u},\quad a.s, \label{XX}
  	\end{equation}
    \begin{equation}
    	Y_{n} = \sum^{T}_{u = 1} X_{nT+u},\quad a.s,\label{2.5}
    \end{equation}
    are stationary process.
  \end{proposition}
   \par According to the above formula and Theorem 4.4.1 in \cite{brockwell1991time}, we deduce the spectral density of $Y_n$ from the spectral density of $\epsilon^H_t$. The proof will be presented in detail in Section 6.
   {\begin{remark}
 The stationary process \((Y_{n})\) encompasses all the information of \(\epsilon^{H}_{n}\). Therefore, we will utilize \((Y_{n})\) to obtain the estimation of \(H\) and the one-step estimator. 
   \end{remark}}
  \begin{proposition}
  	Let $	f_{H,\phi(u)}(\lambda)$ be the spectral density of $(X_{nT+u})_{n \in \mathbb{N}}$, then it can be rewritten as
  	\begin{equation}
  		f_{H,\phi(u)}(\lambda) = (1-2\phi(u)cos\lambda T +\phi^{2}(u))^{-1}	{f_{\epsilon_{n}^{H}}(\lambda)}. \label{XX}
  	\end{equation}
   
  \end{proposition}
  \begin{proposition}
  	Let $g_{H,\underline{{\phi}}}(\lambda)$ be the spectral density of $(Y_{n})_{n \in \mathbb{N}}$, then it can be rewritten as
  	   \begin{equation}
  	   	g_{H,\underline{{\phi}}}(\lambda) = |\sum^{T-1}_{p=0} \Phi_{\phi(T-p)}({\lambda}) |^{2}f_{\epsilon^{H}_{n}}(\lambda), \label{2.7}
  	   \end{equation}
  	   where $ \Phi_{\phi(T-p)}({\lambda}) = \frac{e^{-ip\lambda}}{1-\phi(T-p)e^{-i\lambda T}}$, $p=0,1,...,T-1$.   
   \end{proposition}
  \par Because the GPH estimator is a type of semi-parametric estimation as discussed in \cite{geweke1983estimation}, the explicit expression of $|\sum^{T-1}_{p=0} \Phi_{\phi(T-p)}({\lambda}) |^{2}$ does not affect the estimation of $H$. Thus, the equation $\hat{H}_{n}= {\hat{d}_{n}}+\frac{1}{2}$ remains valid. We can then apply the GPH method directly to the stationary process ${Y_{n}}$. 
  \par Let new series $(Y_{n})_{n \in \mathbb{N}}$ be an observation sample generated via the equation (\ref{2.5}) and choose a suitable integer m which can decrease the mean square error of estimation, where $m < n$. we get the periodogram of ${Y_{n}}$ given by 

  \begin{equation}
  	I(\lambda)=\frac{1}{2\pi n}{|\sum_{t=0}^{n}Y_{t}exp(it\lambda)|^{2}},\label{XX}
  \end{equation}
  \begin{equation}
  	\lambda_{j}=\frac{2\pi j}{n},\quad j{\in}\left \{1,2,...m\right\}, \label{XX}
  \end{equation}
  \begin{equation}
  	a_{j}=\log(2sin\frac{\lambda_{j}}{2}),\quad {\overline {a}_m}=\frac {1}{m}\sum_{j=1}^{m}a_{j},\quad S_{m}=\sum_{j=1}^m(a_j-{\overline {a}_m})^2.\label{XX}
  \end{equation}
  
  \noindent We estimate d by regressing $\log I({\lambda_{j}})$ with respect to $a_{j}$, such that
  \begin{equation}
  	{\hat{d}}_{n} = - \frac{1}{2S_{m}} \sum_{j=1}^m(a_j-{\overline {a}_m})  \log I(\lambda_{j}). \label{XX}
  \end{equation}
  
  \noindent The estimator $\hat{H}_{n}$ is defined by
  \begin{equation}
  	\hat{H}_{n}= {\hat{d}_{n}}+\frac{1}{2}.  \label{GPH}
  \end{equation}

\begin{remark}
  There are several semi-parametric methods for estimating the long memory parameter $d$ and $H$, such as whittle estimation and R/S estimation method proposed by \cite{robinson1995log} and \cite{marinucci1998semiparametric}. These models rely on the log-periodogram approach. However, these methods tend to exhibit greater bias compared to the GPH estimator.

\end{remark}
 
  \subsection{Generalized least squares estimation of SFAR(1) models}	
 \par We now focus on estimating $\phi(u)$ given that the parameter $H$  has been estimated. When the noise in the seasonal autoregressive model is white noise, we can easily obtain the estimator of the parameters of these models using Least Squares Estimation (LSE). However, when the noise is fractional Gaussian noise (fGn), the covariance matrix of fGn is no longer diagonal, making LSE inappropriate. Therefore, we consider using Generalized Least Squares Estimation (GLSE).
  \par To address the effect of seasonal structure on parameter estimation, we apply GLSE to the subsequences $(X_{u},X_{T+u},...X_{nT+u}) $, where $u=1,2,...T$. This allows us to estimate the parameters $\phi(1),\phi(2),\ldots,\phi(T)$ sequentially, assuming the Hurst index is known.

  \par We  deduce the time series can be written in the form 
  \begin{equation}
  	\Phi^{j}_{i}(u)=(X_{u+iT},X_{u+(i+1)T},...X_{u+jT})^{*}, \quad\ i\leq j, \label{XX}
  \end{equation}
  and the autocovariance matrix is  given by
  \begin{equation}
  	\Gamma_{nT,u}(H)=\rho(|i-j|T)_{1\leq i,j\leq n}=
  	\begin{pmatrix}
  		\gamma_{0,u}&\gamma_{T,u}&\gamma_{2T,u}&\cdots&\gamma_{(n-1)T,u}\\
  		\gamma_{T,u}&\gamma_{0,u}&\gamma_{3T,u}&\cdots&\gamma_{(n-2)T,u}\\
  		\vdots &\vdots&\vdots&\ddots&\vdots\\
  		\gamma_{(n-1)T,u}&\gamma_{(n-2)T,u}&\gamma_{(n-3)T,u}&\cdots&\gamma_{0,u}\\
  	\end{pmatrix}, \label{XX}
  \end{equation}
  \noindent where $\gamma_{nT,u}= cov(X_{nT+u},X_{u})$, $n \in \mathbb{N}$.
  We can easily show that $\Gamma_{nT,u}(H)$ depends only on $n$ and not on $u$. Hence, we will denote $\Gamma_{nT,u}$ simply as $\Gamma_{nT}$ from now on, without distinguishing between them.

  \par The estimators $\left\{{\hat{{\phi}}}_{n}(u) \right\}_{u \geq 0}$ are defined by
  \begin{equation}
  	{\hat{{\phi}}}_{n}(u)=\frac{{\Phi^{n}_{2}(u)}^{*}\Gamma^{-1}_{(n-1)T}({\hat{H}}_{n}){\Phi^{n-1}_{1}(u)}}{{\Phi^{n-1}_{1}(u)}^{*}\Gamma^{-1}_{(n-1)T}({\hat{H}}_{n}){\Phi^{n-1}_{1}(u)}}.\label{GLSE}
  \end{equation}

  \par Now, due to the seasonal structure, we need to examine whether the elements of $\Gamma^{-1}_{nT}(\hat{H}_n)$ are finite to assess the feasibility of this method.

  \par Thanks to \cite{fox1986large}, \cite{esstafa2019long}. We know that the elements of $\Gamma^{-1}_{nT}({\hat{H}}_{n})$ can be expressed as a function of the spectral density of fGn. The spectral representation of 
 $(\Gamma^{-1}_{nT})_{j,k}$ implies that 
  \begin{equation}
  	(\Gamma^{-1}_{nT})_{j,k}=\frac{1}{(2\pi)^{2}}\int_{-\pi}^{\pi}\frac{1}{f_{\epsilon_{n}^{H}}(\lambda)}e^{i(k-j)T\lambda}d\lambda. \label{XX}
  \end{equation}
 As $\lambda \rightarrow 0$, according to the definition of fractional Gaussian noise, we have
\begin{equation}
   f_{\epsilon_{n}^{H}}(\lambda) \sim \frac{C_{H}}{2} |\lambda|^{1-2H}, \label{XX}
\end{equation}
{where $C_{H}$ is a constant.} We can categorize the elements of the matrix into two types: diagonal elements and off-diagonal elements.        
 
  \par When $j = k$, we have
\begin{equation}
   (\Gamma^{-1}_{nT})_{j,j} = \frac{1}{(2\pi)^{2}} \int_{-\pi}^{\pi} \frac{1}{f_{\epsilon_{n}^{H}}(\lambda)} \, d\lambda = \frac{1}{2 \pi^{2}} \int_{0}^{\pi} \frac{1}{f_{\epsilon_{n}^{H}}(\lambda)} \, d\lambda. \label{2.18}
\end{equation}
One has when $\lambda \rightarrow 0$ that
\begin{equation}
   \frac{1}{f_{\epsilon_{n}^{H}}(\lambda)} = \frac{2}{C_{H}} |\lambda|^{2H-1} + o\left(\frac{2}{C_{H}} |\lambda|^{2H-1}\right). \label{XX}
\end{equation}
This implies that for $l > 0$ there exists $\delta_{l} > 0$ such that for any $\lambda \in (-\delta_{l}, \delta_{l})$, we have
\begin{equation}
   (1 - l) \frac{2}{C_{H}} |\lambda|^{2H-1} \leq \frac{1}{f_{\epsilon_{n}^{H}}(\lambda)} \leq (1 + l) \frac{2}{C_{H}} |\lambda|^{2H-1}. \label{2.20}
\end{equation}
Thus, equation \eqref{2.18} have an upper bound when $\lambda \in (-\delta_{l}, \delta_{l})$:
\begin{eqnarray}
   |(\Gamma^{-1}_{nT})_{j,j}| 
   &\leq& \frac{1 + l}{C_{H} \pi^{2}} \int_{0}^{\delta_{l}} \lambda^{2H-1} \, d\lambda + \frac{1}{2 \pi^{2}} \int_{\delta_{l}}^{\pi} \frac{1}{f_{\epsilon_{n}^{H}}(\lambda)} \, d\lambda \\
   &\leq& \frac{\delta_{l}^{2H} (1 + l)}{2H C_{H} \pi^2} + \frac{\pi - \delta_{l}}{2 \pi^{2}} \sup_{\lambda \in (\delta_{l}, \pi]} \frac{1}{f_{\epsilon_{n}^{H}}(\lambda)} \nonumber \\
   &\leq& K_{1} \nonumber, \label{XXX}
\end{eqnarray}
where {$K_{1}$} is a constant.
\par When $j \neq k$, according to \cite{esstafa2019long}, there exists a positive constant {$K_{2}$} such that for any $j, k = 1, 2, \ldots$
\begin{equation}
   |(\Gamma^{-1}_{nT})_{j,k}| \leq K_{2} \left| \frac{1}{(k-j)T} \right|^{2H}. \label{XX}  
\end{equation}
Therefore, we have shown that the elements of $\Gamma^{-1}_{(n-1)T}(\hat{H}_{n})$ are finite, which implies that $\hat{\phi}_{n}(u)$ is bounded. From these points we use the notations $H$,$\phi(u)$ and estimator $\hat{H}_{n}(u)$,$\hat{\phi}_{n}(u)$ to present our results concerning the asymptotic properties of the initial estimator.
   \begin{theorem} 
  \par Letting  $m = [n^{\delta}]$ for some {$0<\delta < 1$}, $(X_{nT+u})_{n \geq 0}$ satisfy the equation (\ref{2.1}). Under conditions $(A_{0})$ and $(A_{1})$, we have  
  \begin{center}
   $\binom{\hat{H}_{n}}{\hat{\phi}_{n}(u)} \xrightarrow[n \rightarrow \infty]{{\mathbb{P}}} \binom{H}{\phi(u)},$
\end{center}
for every $u = 1, 2, \ldots, T$, as $n \rightarrow \infty$. Where $\hat{H}_{n}$ and $\hat{\phi}_{n}(u)$ are initial estimators defined in equation (\ref{GPH}) and equation (\ref{GLSE}). {  $[\cdot]$denotes the integer part function.}

\end{theorem} 
   \begin{remark}
    In this proof, we demonstrate that the estimators for each pair of parameters are individually consistent. Consequently, it follows that the estimators for all parameters together are also consistent.
  \end{remark}

  \begin{theorem}
    Let $m =[n^{\delta}]$ for some $\frac{1}{2}<{\delta}<{\frac{2}{3}}$. Under conditions $(A_{0})$ and $(A_{1})$, $\hat{\theta}_{n}$ has a (T+1) dimension limiting normal distribution given by 	 
   $$\sqrt{m}\left(
  \begin{array}{c}
  	\hat{H}_{n}-H\\
  {\hat{\phi}_{n}(1)}-{\phi(1)}\\
  	\vdots\\
  {\hat{\phi}_{n}(T)}-{\phi(T)}
  \end{array}
  \right)\xrightarrow [n\rightarrow \infty]{{{\mathcal{L}}}}  \mathcal{N}(0,\Sigma_{\theta}),
  $$
  where ${\hat{\theta}}_{n} = ({\hat{H}_{n}},{\hat{{\phi}}}_{n}(1),{\hat{{\phi}}}_{n}(2),...{\hat{{\phi}}}_{n}(T))$. The covariance matrix $\Sigma_{\theta}$ is of the form $\Sigma_{\theta} = V_{H}\widetilde{\Sigma}_{\theta}$
  	and $V_{H}$	is the asymptotic variance of $\sqrt{m}({\hat{H}}_{n}-H)$, $\widetilde{\Sigma}_{\theta}$ is a built-in singular matrix.
  \end{theorem}
  \begin{remark} {\citep {hariz2024fast} represents \(\frac{1}{2} < \delta < \frac{2}{3}\), and \cite{hurvich1998mean} states that if \(m = n^{\delta}\), where \(0 < \delta < 1\), it can ensure the asymptotic normality of  \(\hat{H}_{n}\). The condition \(\frac{1}{2}< \delta\) is to ensure that \(\hat{\phi}_{n}(1)\) and \(\hat{\phi}_{n}(2)\) are asymptotically normal. But according to \cite{kutoyants2016multi}, if \(\frac{2}{3} < \delta\), a multi-step estimator may be required, which contradicts our consideration of a one-step estimator. Thus, we consider restricting \(\delta\) to the interval \((\frac{1}{2},\frac{2}{3})\).}
  \end{remark}
  \begin{remark}
  	These results can be extented to the SFAR(p), provided that $\underline{{X}}(u)$ is stationary. 
  \end{remark}
  \begin{remark}
  Additionally, the estimation of $\underline {\phi}$ can be also approached using methods from \cite{brouste2014asymptotic} and \cite{soltane2024asymptotic}.	 
  \end{remark}

\section{One-step estimator of SFAR(1) models}	
  \par In this section, we explore modifications to the initial estimator \(\hat{\theta}_n\) to develop a one-step estimator \(\tilde{\theta}_n\).
  \par We assume that \(Y_n\) is stationary with a spectral density \(g_{H, \underline{\phi}}(\lambda)\), as obtained in equation (\ref{2.7}). For \(g_{H, \underline{\phi}}(\lambda)\) to satisfy the necessary regularity conditions as follow, \\
  $Conditon.1$ For any $\theta =(\theta_{j_{1}},\theta_{j_{2}},\dots,\theta_{j_{T+1}})\in \Theta^{l} $,  where ${\Theta^{l}}$ be an open subset of $R^{T+1}$, $g_{H,\underline {\phi}}(\lambda)$ is three times continuously differentiable on $\Theta^{l} $. In addition, for any $0 \leq k \leq 3$ and $j_{1},...j_{T+1}$, the partial derivative
  \begin{equation}
  	\frac{\partial^{k}}{\partial{\theta_{j_{1}}},...{\partial{\theta_{j_{k}}}}}g_{H,\underline {\phi}}(\lambda),\label{XX}
  \end{equation}
  
  \noindent is a continue equation  on $\Theta^{l} \times [-\pi,\pi]\backslash\left\{0\right\}$, is continuously differentiable with respect to $\lambda$ and its partial derivative
  \begin{equation}
  	\frac{\partial^{k+1}}{\partial\lambda\partial{\theta_{j_{1}}},...{\partial{\theta_{j_{k}}}}}g_{H,\underline {\phi}}(\lambda),  \label{XX}
  \end{equation}
  and is continuous on $\Theta^{l} \times [-\pi,\pi]\backslash\left\{0\right\}$.\\
  $Conditon.2$ There also exists a continuous function $\alpha$: $\Theta^{l}\longrightarrow (-1,1)$, such that for any compact set ${\Theta^{l^{\star}}} \subset \Theta^{l}$ and $\delta>0$, the following conditions hold for every $(\theta, \lambda) \in \Theta^{\star}\times[-\pi,\pi]\backslash\left\{0 \right\}$ are
  \begin{equation}
  	c_{1,\delta,{\Theta^{l}}^{\star}}|\lambda|^{-\alpha(\theta)+\delta} \leq g_{H,\underline {\phi}}(\lambda) \leq c_{2,\delta,{\Theta^{l}}^{\star}}|\lambda|^{-\alpha(\theta)-\delta}, \label{XX}
  \end{equation}  
  \noindent and
  \begin{equation}
  	|\frac{\partial}{\partial \lambda}g_{H,\underline {\phi}}(\lambda)| \leq c_{2,\delta,{\Theta^{l}}^{\star}}|\lambda|^{-\alpha(\theta)-1-\delta},  \label{XX}
  \end{equation}        
  \noindent for any $k \in \left\{1,2,3\right\}$ and any $j \in (1,...,T+1)^{k}$. where 
  \begin{equation}
  	|\frac{\partial^{k}}{\partial{\theta_{j_{1}}},...{\partial{\theta_{j_{k}}}}}g_{H,\underline {\phi}}(\lambda)| \leq c_{2,\delta,\Theta^{l^\star}}|\lambda|^{-\alpha(\theta)-\delta},  
    \label{XX}
  \end{equation} 
   here, $ c_{1,\delta,\Theta^{l^\star}}$ and $ c_{2,\delta,\Theta^{l^\star}}$is some positive finite constant which only depends upon $\delta$ and $\Theta^{l^\star}$. We will prove the spectral density of $Y_{n}$ satisfy regular condition in  auxiliary results.
  
  \begin{proposition}
  	We let $l_{n}$ be the log-likelihood function of a stationary process $(Y_{n})_{n\in \mathbb{N}}$. We assume that $g_{H,\underline {\phi}}(\lambda)$ satisfies the regularity conditions and let $B({\theta, R})$(open ball of center $\theta$ and radius R) for some $R > 0$. For any $t\in B(\theta,R)$, $u \in \mathbb{N}$
  	\begin{equation}
  		l_{n}(\theta + \frac{t}{\sqrt{n}})-l_{n}(\theta) = t\frac{\nabla l_{n}(\theta)}{\sqrt{n}}-\frac{t \mathcal{I}(\theta)t^{*}}{2} + r_{n,\theta}(t),\label{XX} 	
  	\end{equation}
  	\noindent when $n \rightarrow \infty $, the score function         $\nabla(\cdot)$ satisfies 	
    \begin{equation}
  		\frac{\nabla l_{n}(\theta)}{\sqrt{n}}\xrightarrow [n\rightarrow\infty]{\mathbb{P}} \mathcal{N}(0,\mathcal{I}(\theta)) \label{XX}
    \end{equation}
  	\noindent and
    \begin{equation}
  		r_{n,\theta}(t) \xrightarrow [n\rightarrow\infty]{a.s.} 0,  \label{XX}
    \end{equation}
  	\noindent uniformly on each compact set. The Fisher information matrix is given in our case by
  	\begin{equation}
  		\mathcal{I}(\theta) = \frac{1}{4\pi}\left(\int_{-\pi}^{\pi}\frac{\partial log \,   g_{H,\underline {\phi}}(\lambda)}{\partial \theta_{k}}\frac{\partial log \,   g_{H,\underline {\phi}}(\lambda)}{\partial \theta_{j}}d\lambda\right)_{1 \leq k,j \leq T+1}. \label{3.3}
  	\end{equation}
  \end{proposition} 
  This result is a direct consequence of Theorem from \cite{cohen:hal-00638121}.
  \par Since \( g_{H, \underline{\phi}}(\lambda) \) satisfies the regularity conditions, the elements of the Fisher information matrix \(\mathcal{I}(\theta)\) are finite.  After obtaining the Fisher information matrix \(\mathcal{I}(\theta)\) and the log-likelihood function of \((Y_n)_{n \in \mathbb{N}}\), we can compute the one-step estimator as follows 
\begin{equation}
    \tilde{\theta}_n = \hat{\theta}_n + \mathcal{I}(\hat{\theta}_n)^{-1} \frac{1}{n} \nabla l_n(\hat{\theta}_n), \label{3.10}
\end{equation}
We can now analyze the asymptotic properties of the one-step estimator.

  \begin{theorem}
  	  Let ${\hat{\theta}}_{n}$ is the initial estimator of $\theta$, ${\tilde{\theta}}_{n}$ is the one-step estimator of $\theta$. When $g_{H,\underline {\phi}}(\lambda)$ satisfy regular condition, we have a asymptotic normal distribution of ${\tilde{\theta}}_{n}$ that
  	\begin{center}
  		$\sqrt{n}({\tilde{\theta}}_{n}-\theta)\xrightarrow[n\rightarrow\infty]{\mathbb{P}}\mathcal{N}(0,\mathcal{I}(\theta)^{-1}).$
  	\end{center}
  \end{theorem} 
  \begin{remark}
    The parameter \(\theta\) should not lie on the boundary of the parameter space \(\Theta^l\).
\end{remark}
\begin{remark}
   {The one-step estimatior can be applied more generally even if the initial estimator \(\hat{\theta}_n\) does not satisfy asymptotic normality. According to proposition 2.3 in Hariz (2025), if the initial estimator with convergence speed lower than $\sqrt{n}$ and the spectral density of time series meets the regular condition, then the one-step estimator \(\tilde{\theta}_n\) can still achieve asymptotic normality.}
\end{remark}
\begin{remark}
 {One-step estimator can achieve Hájek's lower bound, thus it is  asymptotically efficient in the local minimax sense. We can find related conclusions in \cite{brouste2020one}, \cite{cohen:hal-00638121}.}
\end{remark}

 \section{Simulation study}   
    
    \par According to equation (\ref{2.5}), the likelihood function based on the sample \(\underline{Y}^{(n)} = (Y_{0}, Y_{1}, \ldots, Y_{n-1})\) is given by
\begin{equation}
    l_{n}(\theta) = -\frac{1}{2} \log \det \left( \Gamma^{Y}_{n}(\theta) \right) - \frac{1}{2} \underline{Y}^{(n)^{*}} \Gamma^{Y}_{n}(\theta) \underline{Y}^{(n)}, \label{XX}
\end{equation}
where \(\Gamma^{Y}_{n}(\theta)\) is the covariance matrix of \(\underline{Y}^{(n)}\). For any \(K \in \mathbb{N}\),
\begin{equation}
     Cov (Y_{0}, Y_{k}) = \int_{-\pi}^{\pi} \exp(ik\lambda) \, g_{H, \underline{\phi}}(\lambda) \, d\lambda, \label{XX}
\end{equation}
where \(  Cov(\cdot,\cdot) \) denotes the covariance. The score function with respect to \(\theta\) is given by
\begin{equation}
    \frac{\partial l_{n}(\theta)}{\partial \theta_{i}} = -\frac{1}{2}  Tr   \left( \left( \Gamma^{Y}_{n}(\theta) \right)^{-1} \frac{\partial \Gamma^{Y}_{n}(\theta)}{\partial \theta_{i}} \right) + \frac{1}{2} \underline{Y}^{(n)^{*}} \left( \Gamma^{Y}_{n}(\theta) \right)^{-1} \frac{\partial \Gamma^{Y}_{n}(\theta)}{\partial \theta_{i}} \left( \Gamma^{Y}_{n}(\theta) \right)^{-1} \underline{Y}^{(n)}, \label{XX}
\end{equation}
where \( Tr(\cdot)\) denotes the trace of a matrix. The Fisher information matrix (FIM) can be deduced from equation (\ref{3.3}). We simulate the spectral density and its derivatives using the method described in \cite{hariz2024fast}, then plug the FIM and score functions into equation (\ref{3.10}) to compute the one-step estimator numerically.

{For each set of parameters, specifically \((\phi(1),\phi(2),H)=(0.6,0.2,0.8)\) and \((\phi(1),\phi(2),H)=(0.2,0.8,0.6)\), we conduct \(M = 1000\) Monte Carlo simulations. The sample sizes considered are \(n = 100\), \(n = 1000\), and \(n = 2000\). }The number of Fourier frequencies for the initial estimations is set as \(m = [  n^{0.6}]\) and remains fixed throughout the simulations. 
 Without loss of generality, we assume \(T = 2\), and the spectral density of \(Y_{n}\) in this case is given by
 \begin{equation}
    g_{H, \phi(1), \phi(2)}(\lambda) = \frac{2 + \phi^{2}(1) + \phi^{2}(2) + A \cos \lambda - B \cos 2\lambda - C \cos 3\lambda}{\left(1 - 2 \phi(2) \cos 2\lambda + \phi^{2}(2)\right)\left(1 - 2 \phi(1) \cos 2\lambda + \phi^{2}(1)\right)} f_{\epsilon^{H}_{n}}(\lambda),
\end{equation}
where \(A = 2 + 2 \phi(1) \phi(2) - 2 \phi(1)\), \(B = 2 (\phi(1) + \phi(2))\), and \(C = 2 \phi(2)\).
 
    \begin{figure}[htbp]
    \centering
   \includegraphics[width=1\textwidth]{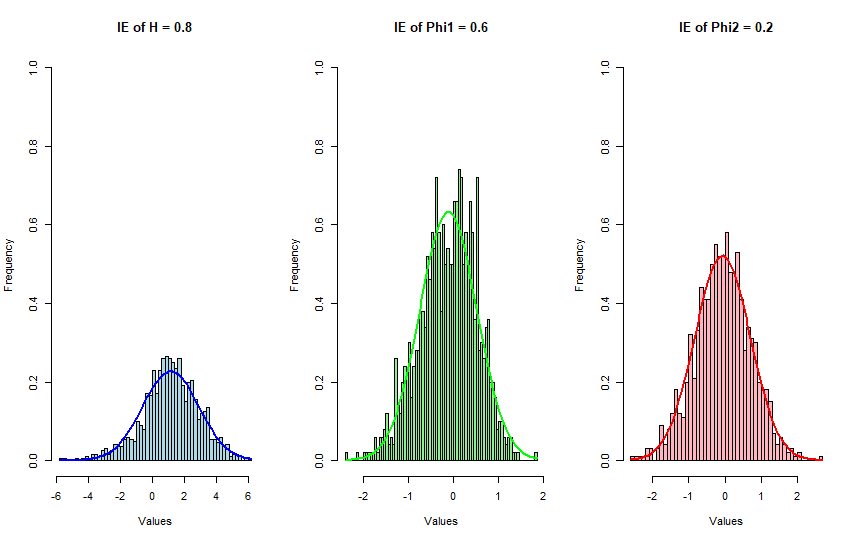}
   \includegraphics[width=1\textwidth]{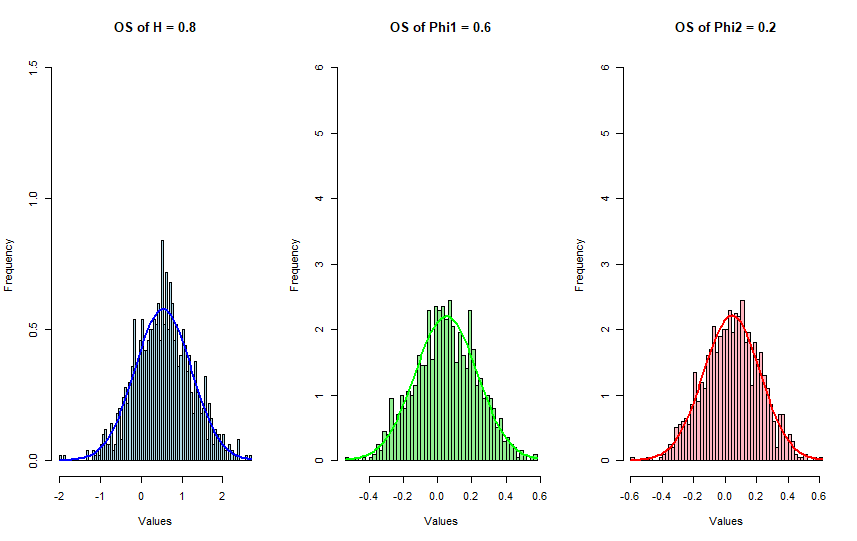}
    \caption{{The simulation of initial estimator and one-step estimator where $\theta=(0.6,0.2,0.8)$ for $m = [n^{\frac{3}{5}}]$, $ n=100$.}}
    \label{628h}
\end{figure}

  \begin{figure}[htbp]
    \centering
   \includegraphics[width=1\textwidth]{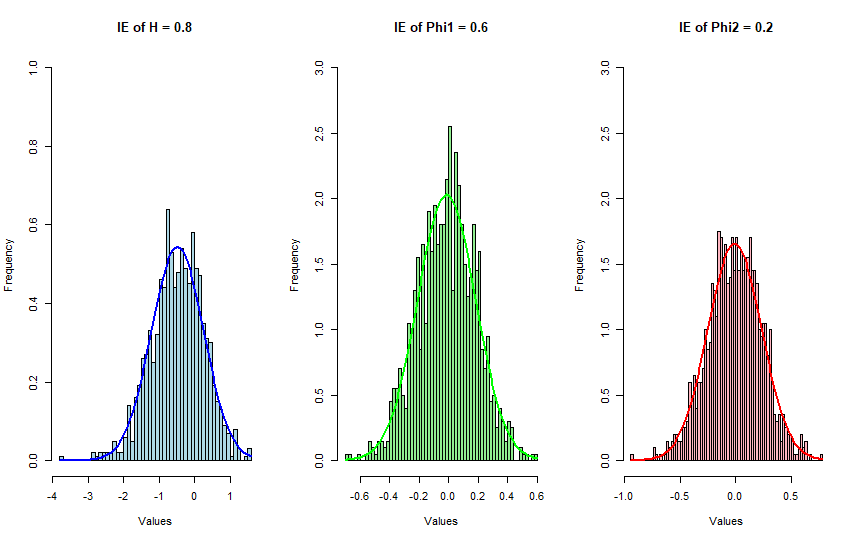}
   \includegraphics[width=1\textwidth]{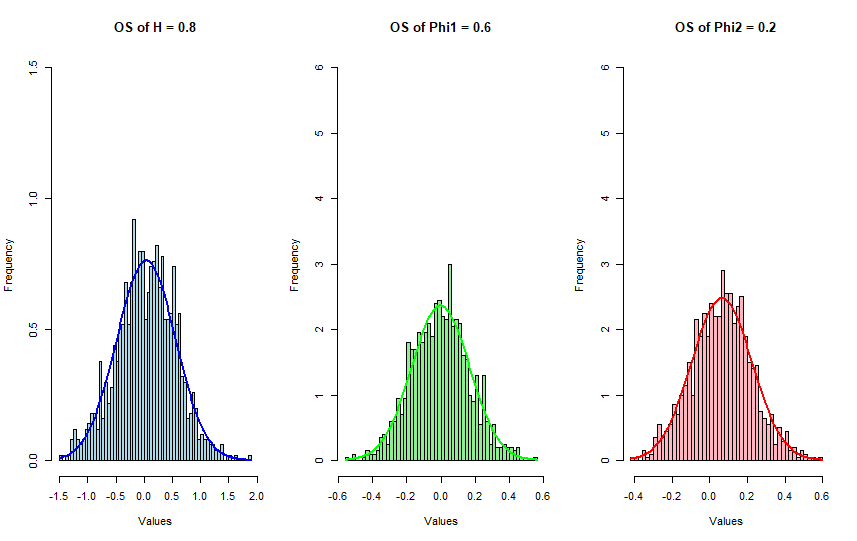}
    \caption{{The simulation of initial estimator and one-step estimator where $\theta=(0.6,0.2,0.8)$ for $m = [n^{\frac{3}{5}}]$,$ n=1000$.}}
    \label{628t}
\end{figure}
 
  \begin{figure}[htbp]
    \centering
   \includegraphics[width=1\textwidth]{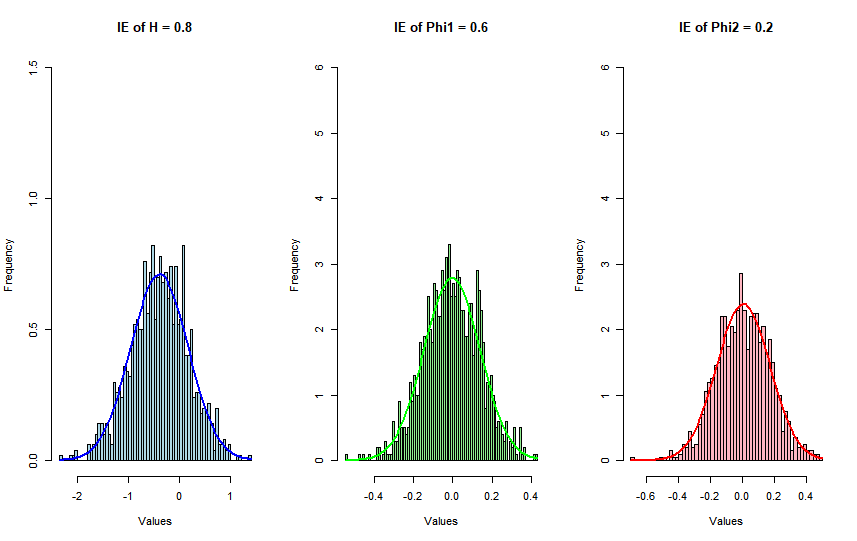}
   \includegraphics[width=1\textwidth]{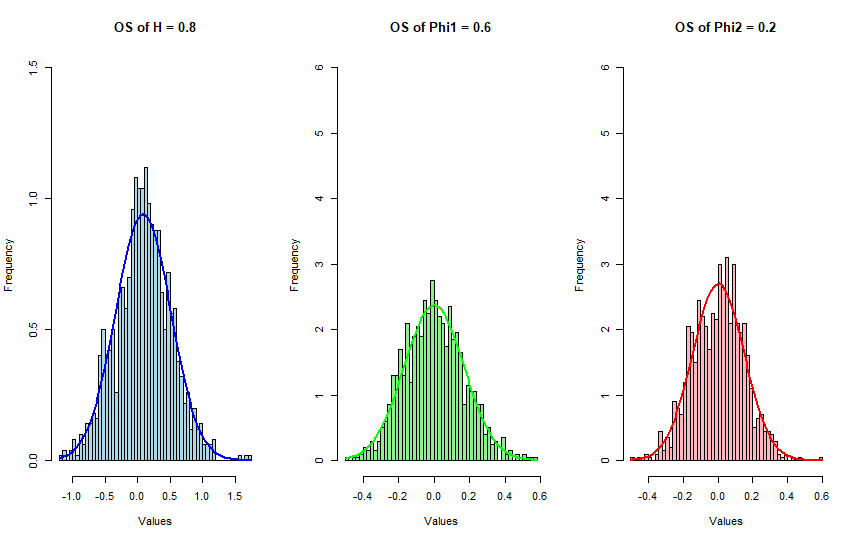}
    \caption{{The simulation of initial estimator and one-step estimator where $\theta=(0.6,0.2,0.8)$ for $m = [n^{\frac{3}{5}}]$,$ n=2000$.}}
    \label{6282t}
\end{figure}

\newpage
\begin{table}[htbp]
    \tbl{{The Bias and RMSE of Initial estimator 
 and One-step estimator for $\theta=(0.6,0.2,0.8)$ when $n = 100$}}
    {\begin{tabular}{|c|c|c|c|c|}
            \hline
            {n = 100} & {B IE} & {B OS} & {RMSE IE} & {RMSE OS} \\
            \hline
            {$H$} & {0.1462} & {0.0683} & {0.2656} & {0.1106} \\
            {$\phi(1)$} & {-0.0143} & {0.0064} & {0.0807} & {0.0237} \\
            {$\phi(2)$} & {-0.0082} & {-0.0055} & {0.0972} & {0.0234} \\
            \hline
        \end{tabular}}
    \label{table1}
\end{table}
\begin{table}[htbp]
    \tbl{{The Bias and RMSE of Initial estimator 
 and One-step estimator for $\theta=(0.6,0.2,0.8)$ when $n = 1000$}}
    {\begin{tabular}{|c|c|c|c|c|}
            \hline
            {n=1000} & {B IE} & {B OS} & {RMSE IE} & {RMSE OS} \\
            \hline
            {$H$} & {-0.0599} & {0.0471} & {0.1101} & {0.0660} \\
            {$\phi(1)$} & {-0.0015} & {-0.0002} & {0.0248} & {0.0212} \\
            {$\phi(2)$} & {-0.0010} & {-0.0081} & {0.0304} & {0.0218} \\
            \hline
        \end{tabular}}
    \label{table2}
\end{table}

\begin{table}[htbp]
    \tbl{{The Bias and RMSE of Initial estimator 
 and One-step estimator for $\theta=(0.6,0.2,0.8)$ when $n = 2000$}}
    {\begin{tabular}{|c|c|c|c|c|}
            \hline
            {n=2000} & {B IE} & {B OS} & {RMSE IE} & {RMSE OS} \\
            \hline
            {$H$} & {-0.0497} & {0.0112} & {0.0864} & {0.0545} \\
            {$\phi(1)$} & {-0.0004} & {0.0001} & {0.0180} & {0.0212} \\
            {$\phi(2)$} & {0.0005} & {0.0003} & {0.0210} & {0.0186} \\
            \hline
        \end{tabular}}
    \label{table3}
\end{table}

     \begin{figure}[htbp]
    \centering
   \includegraphics[width=1\textwidth]{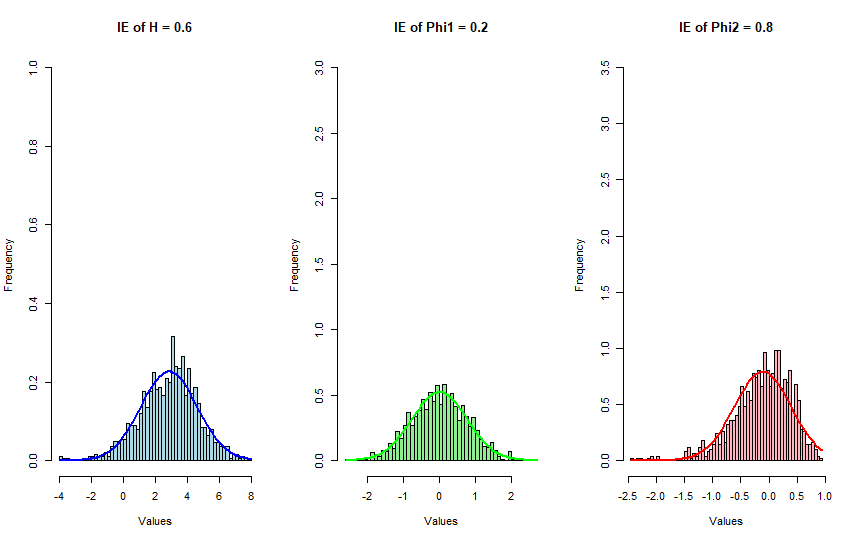}
   \includegraphics[width=1\textwidth]{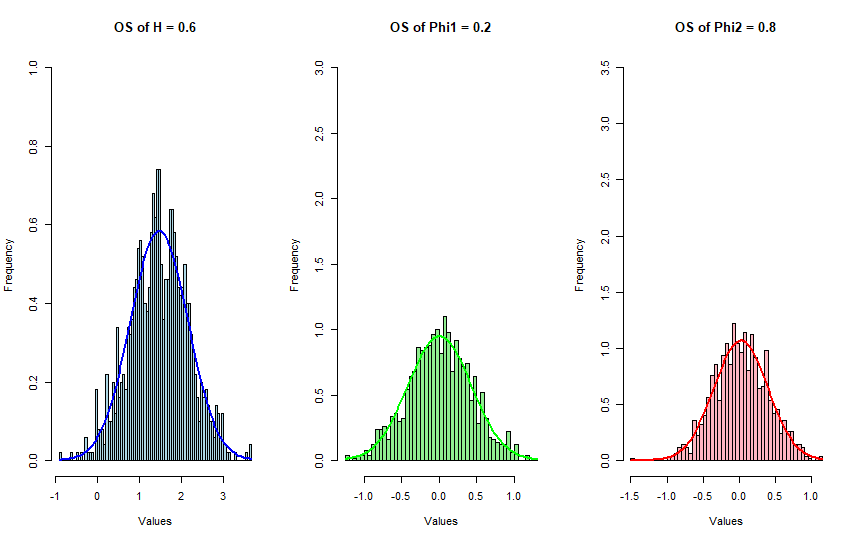}
    \caption{{The simulation of initial estimator and one-step estimator where $\theta=(0.2,0.8,0.6)$ for $m = [n^{\frac{3}{5}}]$, $ n=100$.}}
    \label{286h}
\end{figure}

  \begin{figure}[htbp]
    \centering
   \includegraphics[width=1\textwidth]{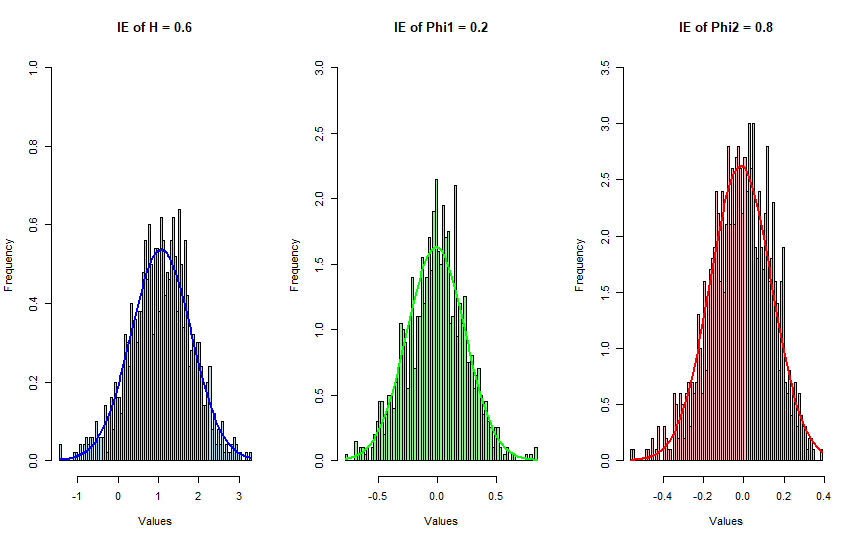}
   \includegraphics[width=1\textwidth]{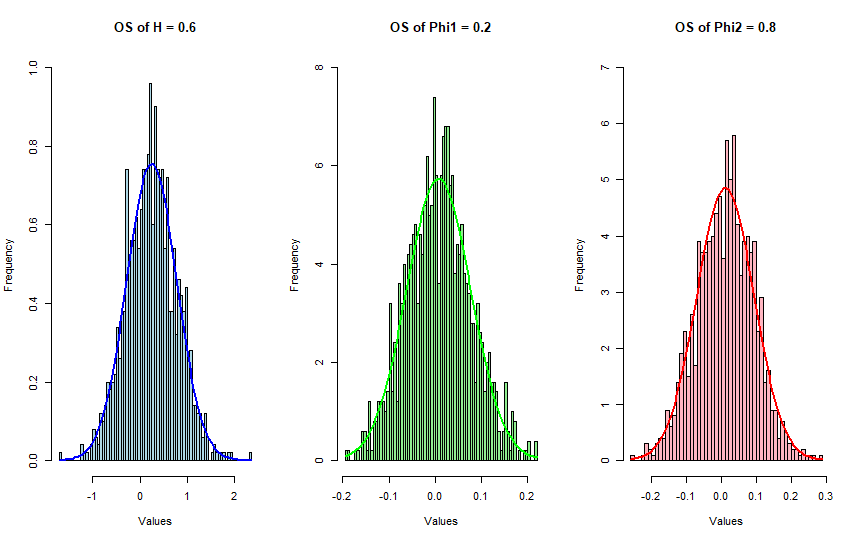}
    \caption{{The simulation of initial estimator and one-step estimator where $\theta=(0.2,0.8,0.6)$ for $m = [n^{\frac{3}{5}}]$,$ n=1000$.}}
    \label{286t}
\end{figure}

  \begin{figure}[htbp]
    \centering
   \includegraphics[width=1\textwidth]{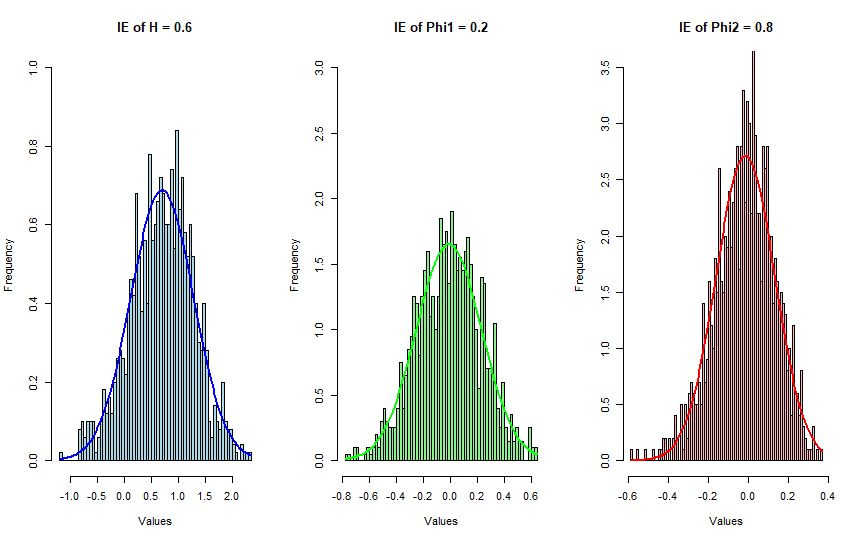}
   \includegraphics[width=1\textwidth]{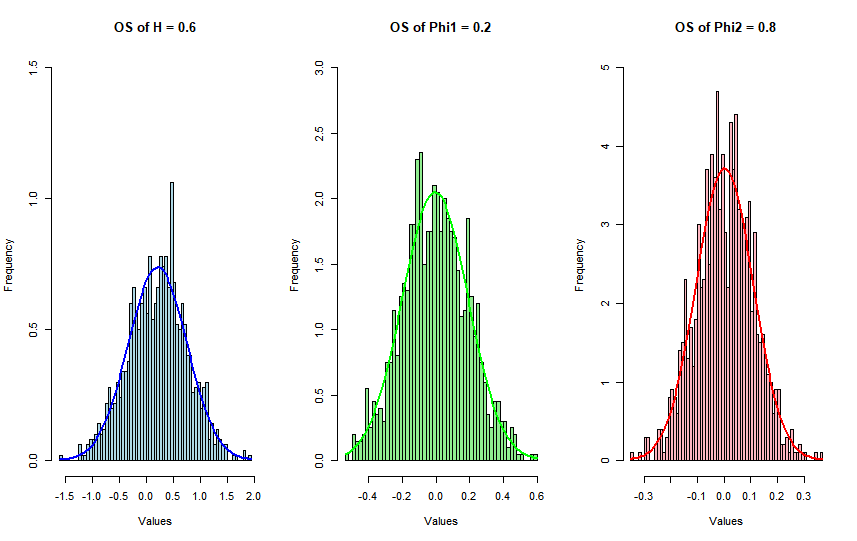}
    \caption{{The simulation of initial estimator and one-step estimator where $\theta=(0.2,0.8,0.6)$ for $m = [n^{\frac{3}{5}}]$,$ n=2000$.}}
    \label{286tt}
\end{figure}
\newpage
\begin{table}[htbp]
    \tbl{{The Bias and RMSE of Initial estimator 
 and One-step estimator for $\theta=(0.2,0.8,0.6)$ when $n = 100$.}}
    {\begin{tabular}{|c|c|c|c|c|}
            \hline
            {n=100} & {B IE} & {B OS} & {RMSE IE} & {RMSE OS} \\
            \hline
            {$H$} & {0.3597} & {0.1846} & {0.4227} & {0.2036} \\
            {$\phi(1)$} & {-0.0009} & {0.0010} & {0.0962} & {0.0530} \\
            {$\phi(2)$} & {-0.0140} & {0.0028} & {0.0651} & {0.0471} \\
            \hline
        \end{tabular}}
    \label{table4}
\end{table}

\begin{table}[htbp]
    \tbl{{The Bias and RMSE of Initial estimator 
 and One-step estimator for $\theta=(0.2,0.8,0.6)$ when $n = 1000$.}}
    {\begin{tabular}{|c|c|c|c|c|}
            \hline
            {n=1000} & {B IE} & {B OS} & {RMSE IE} & {RMSE OS} \\
            \hline
            {$H$} & {0.1351} & {0.0312} & {0.1644} & {0.0736} \\
            {$\phi(1)$} & {-0.0014} & {-0.0009} & {0.0309} & {0.0088} \\
            {$\phi(2)$} & {-0.0019} & {0.0014} & {0.0193} & {0.0104} \\
            \hline
        \end{tabular}}
    \label{table5}
\end{table}

\begin{table}[htbp]
    \tbl{{The Bias and RMSE of Initial estimator 
 and One-step estimator for $\theta=(0.2,0.8,0.6)$ when $n = 2000$.}}
    {\begin{tabular}{|c|c|c|c|c|}
            \hline
            {n=2000} & {B IE} & {B OS} & {RMSE IE} & {RMSE OS} \\
            \hline
            {$H$} & {0.0883} & {0.0274} & {0.1146} & {0.0735} \\
            {$\phi(1)$} & {0.0011} & {0.0006} & {0.0305} & {0.0245} \\
            {$\phi(2)$} & {-0.0032} & {0.0003} & {0.0186} & {0.0135} \\
            \hline
        \end{tabular}}
    \label{table6}
\end{table}
\par {Figure $\ref{628h}$ ,\ref{628t} and \ref{6282t} depict the frequency distribution of statistical errors for the initial estimatior and one-step estimatior of the SFAR(1) model with parameters \(\phi(1) = 0.6\), \(\phi(2) = 0.2\), and \(H = 0.8\). Figure $\ref{286h}$ ,\ref{286t} and \ref{286tt} depict the frequency distribution of statistical errors for the initial estimatior and one-step estimatior of the SFAR(1) model with parameters \(\phi(1) = 0.2\), \(\phi(2) = 0.8\), and \(H = 0.6\).} 
\par {In all the tables, B stands for Bias, IE represents initial estimator, and OS denotes one-step estimator. From the above tables, it can be seen that the OS estimator shows a significant improvement in the estimation of $H$. From these figures and the accompanying table, it is evident that the one-step estimatior outperforms the initial estimation, with a particularly notable improvement in estimating the parameter $H$, at the same time, we found that as the sample size increases, the estimation becomes more efficient. According to \cite{hariz2024fast} and our simulations, the one-step estimation also has a faster running speed.} 
\section{Conclusions and perspectives} 
  \par In this paper, we propose a simple and effective method for estimating the parameters of the SFAR model individually, and we derive the asymptotic properties of this method. We address the difficulty of parameter estimation caused by the non-stationarity of the model by creating new subseries and obtaining an explicit form for the spectral density of the additive series.
 \par The one-step procedure is essentially a gradient descent approach, achieving the \(\sqrt{n}\) rate with optimal variance. 
 \par Our results can be extended to SARIMA models by adjusting the calculation of the covariance matrix of the noise and the spectral density of \(X_n\). Additionally, more effective initial estimators can be utilized for the one-step procedure, similar to the approach taken by \cite{hariz2024fast} in the estimation of FARIMA models.
 \par An interesting aspect to consider is that when \(T\) is sufficiently large, even larger than \(n\), but still finite, the effectiveness of this gradient descent approach may diminish. In such cases, alternative methods for optimizing the initial estimator should be explored.
\section{{Application on Real Data}}
{In this section, we will conduct practical modeling and analysis to examine the application effectiveness of the seasonal autoregressive model driven by fractional Gaussian noise in real data.}
\par {The data on the Colorado River runoff in Arizona selected in this paper are from the public data of the United States Geological Survey. This dataset records the monthly river runoff of the Colorado River from 1922 to 2022, with the unit of cubic feet per second. To facilitate modeling, we average the data of each of the 12 months on a quarterly basis, obtaining the quarterly runoff data for the first, second, third, and fourth quarters respectively. First, we calculate and obtain the autocorrelation function (ACF) plot and partial autocorrelation function (PACF) plot of the quarterly runoff data as follows:}
\begin{figure}[h]
    \centering
\includegraphics[width = 1\textwidth]{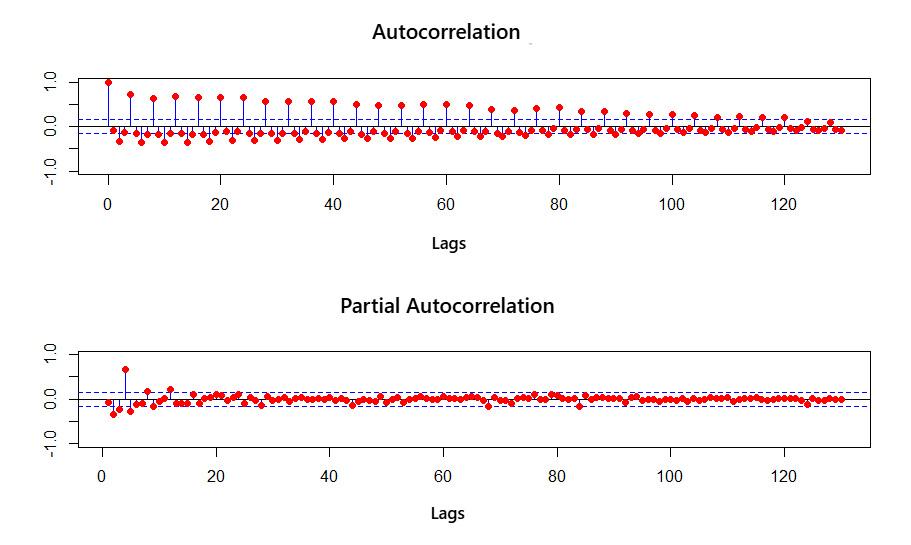}
    \caption{{Autocorrelation and partial autocorrelation coefficient plots of the Colorado River runoff in Arizona}}
    \label{FBMPIC.png}
\end{figure}
\par {From the above figure, we can observe that the autocorrelation function shows a trailing pattern with a slow decay rate, while the partial autocorrelation function cuts off. Therefore, it is appropriate to consider using a fractional AR model with long   memory properties.}

{To avoid data over crowding and considering that the river runoff around 1963 changed significantly for unknown reasons, we extract the runoff data from 1922 to 1962 and draw the following sample path plot.}
\begin{figure}[h]
    \centering
    \includegraphics[width = 0.7\textwidth]{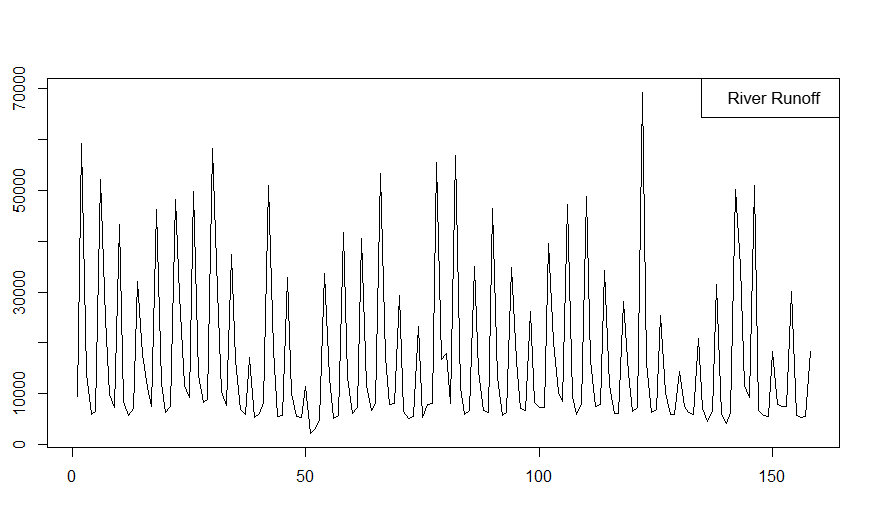}
    \caption{{Sample path plot of the Colorado River runoff in Arizona from 1922 to 1962}}
    \label{path.png}
\end{figure}
\newpage
{As can be seen from Figure \ref{path.png}, the runoff of this river exhibits obvious seasonality. Considering the above two points, in this empirical analysis, we consider using the SFAR(1) model to simulate the above observations and compare it with the simulation of the seasonal autoregressive model driven by white noise.}

\par {The seasonal autoregressive model driven by white noise is as follows:}
{\begin{equation}
    X_{4n + u}=\alpha(4n + u)X_{4(n - 1)+u}+\epsilon_{4n + u},\quad u = 1,2,3,4,\quad n\in\mathbb{Z}^{+},
\end{equation}}
{where $\epsilon_{4n + u}$ is white noise, and $\alpha(1)$, $\alpha(2)$, $\alpha(3)$, $\alpha(4)$ are the model coefficients, satisfying $\alpha(4n + u)=\alpha(u)$.}
\par {We utilized the data  from 1922 to 1962 to derive the parameter estimations of the two models. Subsequently, we computed their RMSEb and MAE against the real data. Finally, we randomly simulated 20 data points within these forty years using these two models.  The results are shown in the following table and figure.}
 \begin{table}[!htbp] %voc table result
    \centering
    \caption{{Fitting results of different models}}
    \begin{tabular}{*{6}{c}}
        \toprule
        Model & SFAR & SAR \\
        \midrule
        Parameters & ($\phi(1)$, $\phi(2)$, $\phi(3)$, $\phi(4)$, $H$) & ($\alpha(1)$, $\alpha(2)$, $\alpha(3)$, $\alpha(4)$) \\
        Values & (0.96, 0.82, 0.80, 0.90, 0.60) & (0.80, 0.61, 0.16, 0.53) \\
        RMSE & 9439.37 & 16773.58 \\
        MAE & 6107.16 & 12264.90 \\
        \bottomrule
    \end{tabular}
    \label{table1}
\end{table} 

\newpage
\begin{figure}[h]
    \centering
    \includegraphics[width = 0.8\textwidth]{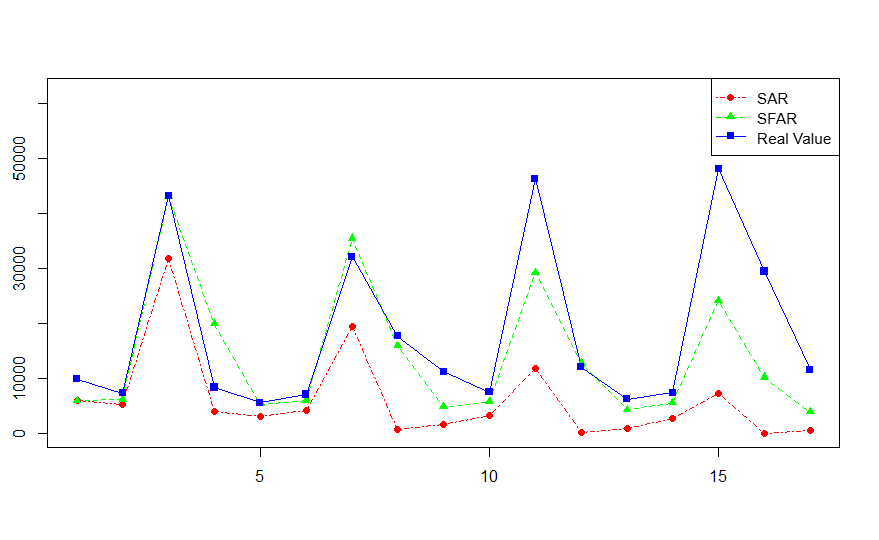}
    \caption{{Model fitting plot}}
    \label{model.png}
\end{figure}
{In the above figure, the values represents the parameters fitted by the SFAR and SAR model. From Figure \ref{model.png} and Table \ref{table1}, we can see that the seasonal autoregressive model driven by fractional noise has smaller RMSE and MAE values and better fitting performance. Therefore, the seasonal autoregressive model with long memory properties is more suitable for the study of the Colorado River runoff.}
  \section{Proofs of the main results} 	
  \par For clarity, we divide the technical results into two parts. {The first part addresses the stationarity and spectral density of the SFAR(1) model, as well as the asymptotic properties of the initial estimator. The second part focuses on the asymptotic properties related to the one-step estimator.}
  \subsection{Proof of Proposition 2.1} 
  \par We will utilize the following lemma to demonstrate the stationarity of $Y_{n}$. 
  \begin{lemma}
  	For any $  u \in \mathbb{N}$, the SFAR(1) model is defined by the recursive scheme 
   \begin{equation}
    X_{nT+u}=\phi(u)X_{(n-1)T+u}+\epsilon^{H}_{nT+u} ,
   \end{equation}
   where $\epsilon^{H}_{nT+u}$ is a fractional Gaussian noise, $\phi(u) = \phi(u+nT) \leq 1$, $n \in \mathbb{N}$, then $(X_{nT+u})_{n \in \mathbb{N}}$ is a stationary process. 
  \end{lemma} 
 \begin{proof} We verify that the process satisfies the three conditions for weak stationarity individually.\\
   (1) For any $ u \in \mathbb{N}$, $E(X_{nT+u}) = \mu $ is a finite constant.
  \par Because the equation \((1 - \phi(u))z = 0\) has a root outside the unit circle, the process \((X_{nT+u})_{n \in \mathbb{N}}\) is said to be an SFAR(1) process if it can be represented as follows
\begin{equation}
    X_{nT+u} = \sum_{j=0}^{\infty} \phi^{j}(u) \epsilon^{H}_{(n-j)T+u}, \label{lemma1.1}
\end{equation}
 without loss of generality, we assume $E(\epsilon^{H}_{n}) = 0$. For any time series \((X_{nT+u})_{n \in \mathbb{N}}\) under the monotone convergence theorem and the Cauchy-Schwarz inequality, we obtain
  \begin{eqnarray}
  	E|X_{nT+u}| 
  	&\leq&
  	E\sum_{j=0}^{\infty}|\phi^{j}(u)\epsilon^{H}_{(n-j)T+u}|=\sum_{j=0}^{\infty}E|\phi^{j}(u)\epsilon^{H}_{(n-j)T+u}| \label{6.2} \\
  	&\leq &
  	\sum^{\infty}_{j=0}|\phi^{j}(u)|E|\epsilon^{H}_{(n-j)T+u}| \nonumber\\
  	&\leq &
  	C \sum_{j=0}^{\infty}|\phi^{j}(u)|. \nonumber 
  \end{eqnarray}
  \par We know that \( \phi^j(u) = o(\rho^{-j}) \) as \( j \rightarrow \infty \) for \( 1 < \rho < \frac{1}{\phi(u)} \). Consequently, \( \phi^j(u) \) is absolutely summable, i.e., \( \sum_{j=0}^{\infty} |\phi^j(u)| < \infty \). Thus, \( E|X_{nT+u}| < \infty \) as shown in equation (\ref{6.2}). By the monotone convergence theorem, \( \sum_{j=0}^{\infty} \phi^j(u) \) is absolutely convergent almost surely.

\par Considering that
\begin{equation}
    \left| \sum_{j=0}^{k} \phi^j(u) \epsilon^{H}_{(n-j)T+u} \right| \leq \sum_{j=0}^{k} \left| \phi^j(u) \epsilon^{H}_{(n-j)T+u} \right|
\end{equation}
with the dominated convergence theorem, 
\begin{equation}
    E(X_{nT+u}) = \lim_{k \to \infty} E \left( \sum_{j=0}^{k} \phi^j(u) \epsilon^{H}_{(n-j)T+u} \right) = 0. \label{XX}
\end{equation}
  (2) For any $ u \in \mathbb{N}$, $E(X_{nT+u})^{2} \leq \infty $.
\par From equation \eqref{lemma1.1}, we derive
\begin{eqnarray}
    E(X_{nT+u})^2 
    &=& E \left( \sum_{j=0}^{\infty} \phi^j(u) \epsilon^{H}_{(n-j)T+u} \right)^2 \\
    &=& E \left( \sum_{s=0}^{\infty} \sum_{k=0}^{\infty} \phi^s(u) \phi^k(u) \epsilon^{H}_{(n-s)T+u} \epsilon^{H}_{(n-k)T+u} \right) \nonumber, \label{6.5}
\end{eqnarray}
 by applying the conclusion above, we obtain
\begin{equation}
    E \left| X_{nT+u} \right|^2 = \sum_{s=0}^{\infty} \sum_{k=0}^{\infty} |\phi^s(u) \phi^k(u)| E \left| \epsilon^{H}_{(n-s)T+u} \epsilon^{H}_{(n-k)T+u} \right|, \label{6.6666}
 \end{equation}
and the covariance of \( \epsilon^{H}_{(n-s)T+u} \) and \( \epsilon^{H}_{(n-k)T+u} \) is
\begin{equation}
    \rho(\epsilon^{H}_{(n-s)T+u}, \epsilon^{H}_{(n-k)T+u}) = \frac{1}{2} \left( |(s-k)+1|^{2H} - 2 |(s-k)T|^{2H} + |(s-k)T-1|^{2H} \right).
\end{equation}

\par Since \( \phi^j(u) \) is absolutely summable, it is also square summable. Additionally, as \( s - k \to \infty \), \( \rho(\epsilon^{H}_{(n-s)T+u}, \epsilon^{H}_{(n-k)T+u}) \to 0 \), implying that there exists a constant \( M \) such that \( E \left| \epsilon^{H}_{(n-s)T+u} \epsilon^{H}_{(n-k)T+u} \right| \leq M \). Based on the above discussion and equation (\ref{6.5}), we have established that \( E(X_{nT+u})^2 \leq \infty \).\\
(3) For any $ k,s \in \mathbb{N}$, $E(X_{kT+u}-\mu)E(X_{sT+u}-\mu) = \gamma_{(k-s)T}$, which means that the autocovariance of $X_{kT+u}$ and $X_{sT+u}$ depends only on the time interval $(k-s)T$.  
  \par Without losing of generality, we assume $\mu = 0$ and the covariance of  $(X_{nT+u})_{n \in \mathbb{N}}$ be rewritten as
    \begin{align}
      E(X_{kT+u}X_{sT+u})
      & = 
	 E(\sum^{\infty}_{i=0}\sum^{\infty}_{j=0} \phi^{i}(u) \phi^{j}(u)\epsilon^{H}_{(s-i)T+u}\epsilon^{H}_{(k-j)T+u})   \\
     & \quad =\sum^{\infty}_{i=0}\sum^{\infty}_{j=0} \phi^{i}(u) \phi^{j}(u)\frac{1}{2}(|(s-k+j-i)T+1|^{2H} \nonumber \\
    & \qquad-
	 2|(s-k+j-i)T|^{2H}+|(s-k+j-i)T-1|^{2H}), \nonumber
  \end{align}
 for any $ q \in \mathbb{N}$, it follows directly from the above equation that 
 \begin{equation}
 E(X_{(k+q)T+u} X_{(s+q)T+u}) = E(X_{kT+u} X_{sT+u}).
 \end{equation}
  \par Thus, we have shown that $ (X_{nT+u} )_{n \in \mathbb{N}}$ is stationary. Since $Y_{n}$ is a combination of $X_{nT+u}$  in a cyclic manner, its stationarity naturally follows.
\end{proof}
  \subsection{Proof of Proposition 2.2}
      Since the stationarity of \( (X_{nT+u})_{n \in \mathbb{N}} \) has been established, we can determine its spectral density using Theorem 4.4 \citep{brockwell1991time}. $(X_{nT+u})_{n \in \mathbb{N}}$ satisfies the recursion
\begin{eqnarray}
    X_{nT+u} 
    &=& X_{(n-1)T+u} + \epsilon^{H}_{nT+u} \\
    &=& \sum_{j=0}^{\infty} \phi^j(u) \epsilon^{H}_{(n-j)T+u} \nonumber \\
    &=& \sum_{j=0,T,2T,\ldots} \phi^{\frac{j}{T}}(u) \epsilon^{H}_{nT+u-j}.\nonumber
\end{eqnarray}
From the above expression, we obtain the transfer function \( H(e^{-i\lambda}) \) as follows 
\begin{equation}
    H(e^{-i\lambda}) = \sum_{k=0}^{\infty} \phi^k(u) e^{-i\lambda kT}.
\end{equation}
Thus, the spectral density function \( f_{H,\phi(u)}(\lambda) \) is given by 
\begin{equation}
    f_{H,\phi(u)}(\lambda) = |H(e^{-i\lambda})|^2 f_{\epsilon^{H}_{n}}(\lambda) = \frac{f_{\epsilon^{H}_{n}}(\lambda)}{1 - 2\phi(u) \cos(\lambda T) + \phi^2(u)}.
\end{equation}

  \subsection{Proof of Proposition 2.3}
     Because the stationarlity of $(Y_{n})_{n \in \mathbb{N}}$  has been proved above, $Y_{n}$ has the following expression
   \begin{eqnarray}
   	Y_{n}
   	& = &  \sum^{T}_{u=1}X_{nT+u}     \\
   	& = & \sum^{\infty}_{j=0}\phi^{j}(1)\epsilon^{H}_{(n-j)T+1}+\sum^{\infty}_{j=0}\phi^{j}(2)\epsilon^{H}_{(n-j)T+2}+\cdots+\sum^{\infty}_{j=0}\phi^{j}(T)\epsilon^{H}_{(n-j)T+T} \nonumber   \\ 
   	& = &
   	\sum^{\infty}_{k=0} \tilde{h}_{k}\epsilon^{H}_{(n+1)T-k}, \nonumber    
   \end{eqnarray}
   and the coefficient of transfer function has the form of 
   \begin{equation}
   \tilde{h}_{k} = \left\{
   \begin{aligned}
   	&\phi^{\frac{k}{T}}(T) && \text{if } k=0,T,2T,...\\
   	&\phi^{\frac{k-1}{T}}(T-1) && \text{if } k= 1,1+T,1+2T,...\\
   	& \qquad \vdots & \\
   	&\phi^{\frac{k-(T-1)}{T}}(1) && \text{if } k=T-1,2T-1,3T-1,... \\	
   \end{aligned}
   \right.
   \end{equation}
  Then the transfer function is given by 
   \begin{equation}
      H(e^{-i\lambda}) =  \sum^{\infty}_{j=kT}\phi^{\frac{j}{T}}(T)e^{-ij\lambda}+\sum^{\infty}_{j=1+kT}\phi^{\frac{j-1}{T}}(T-1)e^{-ij\lambda}+\cdots+\sum^{\infty}_{j=(k+1)T-1}\phi^{\frac{j-T+1}{T}}(1)e^{-ij\lambda}.
   \end{equation}
   To simplify the notation, we denote
   \begin{equation}
     \Phi_{\phi(T-p)}({\lambda}) = \sum^{\infty}_{j=kT+p}\phi^{\frac{j-p}{T}}(T-p)e^{-ij\lambda} = \frac{e^{-ip\lambda}}{1-\phi(T-p)e^{-i\lambda T}},\quad p=0,1...T-1.
   \end{equation}
   Then, we have
   \begin{equation}
     H(e^{-i\lambda}) = \left(\sum^{T-1}_{P=0} \Phi_{\phi(T-p)}({\lambda})  \right),
  \end{equation}
   and the spectral density of $(Y_{n})_{n \in \mathbb{N}}$ has the representation
   \begin{equation}
   g_{H,\underline{{\phi}}}= |H(e^{-i\lambda})|^{2}f^{H}_{\epsilon^{H}_{n}}(\lambda)=|\sum^{T-1}_{P=0} \Phi_{\phi(T-p)}({\lambda})|^{2}f^{H}_{\epsilon^{H}_{n}}(\lambda).
   \end{equation}
   \begin{remark}
   	$|\cdot|$ denote the modulus of $H(e^{-i\lambda})$.
   \end{remark}
  \subsection{Proof of Theorem 2.4} 
  \par The first part of the proof to to establish the consistency of $\hat H_{n}$, while second part is to verify consistency of $\hat\phi_{n}(u)$.\\
(1) Consistency of $\hat H_{n}$ 
  \par This proof is based on Lemma 5.5 \citep{hariz2024fast} and the corollary \citep{hurvich1998mean}. We can express $f_{H,\phi(u)}$  
  in the following form 
  \begin{eqnarray}
  	\tilde{f}_{H,\phi(u)} (\lambda) 
  	&=& 
  	(1-cos{\lambda T})^{2d}f_{H,\phi(u)} (\lambda)  \\ 
  	&=&
  	(1-cos{\lambda T})^{2d}(1-2\phi(u)cos\lambda T +\phi^{2}(u))^{-1}	{f_{\epsilon_{n}^{H}}(\lambda)} \nonumber \\ 
       &=&
       C_{H}	(1-cos{\lambda T})^{2d+1} (1-2\phi(u)cos\lambda T +\phi^{2}(u))^{-1}   \sum\limits_{j \in Z}\frac{1}{|\lambda+2j\pi|^{2H+1}}.\nonumber
  \end{eqnarray} 
 According to \cite{hurvich1998mean}, it can be concluded that
  \begin{equation}
  	\hat d_{n}-d = -\frac{1}{2S_{m}} \sum^{m}_{j=1}(a_{j} - \overline a_{m})log \;(\tilde{f}_{H,\phi(u)}) - \frac{1}{2S_{m}} \sum^{m}_{j=1}(a_{j}-\overline a_{m})\epsilon_{j}, \label {XX}   	            
  \end{equation}
 where $\epsilon_{j}$ is the error defined in Equation (3) of \cite{hurvich1998mean}. According to the theorem 1 from the aforementioned sources, we have
  \begin{equation}
     \hat d_{n} \xrightarrow[n\rightarrow\infty]{\mathbb{P}}d.
  \end{equation}
 Hence, it is evident that
   \begin{equation}
     \hat H_{n} \xrightarrow[n\rightarrow\infty]{\mathbb{P}}H.
  \end{equation}
(2) Consistency of $\hat\phi_{n}(u)$
  \par Assuming 
   \begin{equation}
   \tilde \Phi^{j}_{i}(u) = (\epsilon^{H}_{u+iT},\epsilon^{H}_{u+(i+1)T},...\epsilon^{H}_{u+jT}), i \leq j,
   \end{equation}
  we can derive the following expression
  \begin{equation}
  	\hat \phi_{n}(u) - \phi (u) = \frac{\tilde \Phi^{n^{*}}_{2}(u) \Gamma^{-1}_{(n-1)T}(\hat H_{n}) \Phi^{n-1}_{1}(u)}{\Phi^{{n-1}^{*}}_{1}(u)\Gamma^{-1}_{(n-1)T}(\hat H_{n}) \Phi^{n-1}_{1}(u)}  .\label{6.8}
  \end{equation}
 We apply the taylor expansion of the matrix $\Gamma^{-1}_{(n-1)T}(\hat H_{n})$ at $H$ to the the numerator, yielding
    \begin{align}
     	\tilde \Phi^{n^{*}}_{2}(u) \Gamma^{-1}_{(n-1)T}(\hat H_{n}) \Phi^{n-1}_{1}(u)
     	&= 
     	\tilde \Phi^{n^{*}}_{2}(u) \Gamma^{-1}_{(n-1)T}(H) \Phi^{n-1}_{1}(u)  \\
     	&\quad + 
     	\tilde \Phi^{n^{*}}_{2}(u) A^{(1)}_{nT}(H)\Phi^{n-1}_{1}(u)(\hat H_{n} - H) \nonumber \\
     	& \quad + 
     	\frac{1}{2} \tilde \Phi^{n^{*}}_{2}(u) A^{(2)}_{nT}(H)\Phi^{n-1}_{1}(u)(\hat H_{n} - H)^{2} \nonumber \\
     	& \quad + 
     	\frac{1}{6}\tilde \Phi^{n^{*}}_{2}(u) A^{(3)}_{nT}(\overline{H}_{n})\Phi^{n-1}_{1}(u)(\hat H_{n} - H)^{3},
     	\nonumber \label{XX}   
     \end{align}    
  thanks to the work of \cite{hariz2024fast}, we have the following three conclusions
  \begin{equation}
  	\frac{1}{n}\tilde \Phi^{n^{*}}_{2}(u) A^{(1)}_{nT}(H) \Phi^{n-1}_{1}(u) \xrightarrow[n\rightarrow\infty]{\mathbb{P}} k^{(1)}_{H,\phi(u)},\label{6.10}
  \end{equation}
  \begin{equation}
  	\frac{1}{n}\tilde \Phi^{n^{*}}_{2}(u) A^{(2)}_{nT}(H) \Phi^{n-1}_{1}(u) \xrightarrow[n\rightarrow\infty]{\mathbb{P}} k^{(2)}_{H,\phi(u)},\label{6.11}
  \end{equation} 
  \begin{equation}
  	n^{-\frac{3}{2}}\tilde \Phi^{n^{*}}_{2}(u) A^{(3)}_{nT}(\overline H_{n}) \Phi^{n-1}_{1}(u) = O_{\mathbb{P}}(1), \label{6.12}
  \end{equation}
    where $k^{(1)}_{H,\phi(u)}$, $k^{(2)}_{H,\phi(u)}$ are constants, $\overline H_{n} \in B(H,|\hat{H}_{n}-H|)$, and $A^{(1)}_{nT}(H)$, $A^{(2}_{nT}(H)$, $A^{(3)}_{nT}(H)$ are
  \begin{eqnarray}
  	A^{(1)}_{nT}(H) = -\Gamma^{-1}_{(n-1)T}(H) \frac{\partial{\Gamma_{(n-1)T}(H)}}{\partial H}\Gamma^{-1}_{(n-1)T}(H),	 
  \end{eqnarray}
  \begin{align}
     	A^{(2)}_{nT}(H) 
     	&=  \Gamma^{-1}_{(n-1)T}(H) \frac{\partial^{2}{\Gamma_{(n-1)T}(H)}}{\partial^{2} H}\Gamma^{-1}_{(n-1)T}(H)   \\
     	& \quad + 2\Gamma^{-1}_{(n-1)T}(H) \frac{\partial{\Gamma_{(n-1)T}(H)}}{\partial H}\Gamma^{-1}_{(n-1)T}(H)\frac{\partial{\Gamma_{(n-1)T}(H)}}{\partial H}\Gamma^{-1}_{(n-1)T}(H), \nonumber
     \end{align}
     \begin{align}
    A^{(3)}_{nT}(H)
    &= 
    -\Gamma^{-1}_{(n-1)T}(H)\frac{\partial^{3}{\Gamma_{(n-1)T}(H)}}{\partial^{3} H}\Gamma^{-1}_{(n-1)T}(H) \\
    &\quad- 
    3\Gamma^{-1}_{(n-1)T}(H)\frac{\partial{\Gamma_{(n-1)T}(H)}}{\partial H} \Gamma^{-1}_{(n-1)T}(H)\frac{\partial^{2}{\Gamma_{(n-1)T}(H)}}{\partial^{2} H}\Gamma^{-1}_{(n-1)T}(H)  \nonumber \\
    &\quad- 
    3\Gamma^{-1}_{(n-1)T}(H) \frac{\partial^{2}{\Gamma_{(n-1)T}(H)}}{\partial^{2} H}\Gamma^{-1}_{(n-1)T}(H)\frac{\partial{\Gamma_{(n-1)T}(H)}}{\partial H}\Gamma^{-1}_{(n-1)T}(H) \nonumber \\
    &\quad- 
    6\Gamma^{-1}_{(n-1)T}(H) \frac{\partial{\Gamma_{(n-1)T}(H)}}{\partial H}\Gamma^{-1}_{(n-1)T}(H)\frac{\partial{\Gamma_{(n-1)T}(H)}}{\partial H} \nonumber \\
    &\quad\quad \times \Gamma^{-1}_{(n-1)T}(H)\frac{\partial{\Gamma_{(n-1)T}(H)}}{\partial H}\Gamma^{-1}_{(n-1)T}(H).  \nonumber 
\end{align}

  \par It has been demonstrated in \cite{esstafa2019long} that
  \begin{equation}
  	\frac{1}{n} \tilde \Phi^{n^{*}}_{2}(u) \Gamma^{-1}_{(n-1)T}( H ) \Phi^{n-1}_{1}(u) \xrightarrow[n\rightarrow\infty]{\mathbb{P}}0, \label{6.13} 
  \end{equation}
  the combination of equations (\ref{6.10}), (\ref{6.11}), 
 (\ref{6.12}), and (\ref{6.13}) allows us to deduce that
  \begin{equation}
  	\frac{1}{n} \tilde \Phi^{n^{*}}_{2}(u) \Gamma^{-1}_{(n-1)T}(\hat H_{n}) \Phi^{n-1}_{1}(u) \xrightarrow[n\rightarrow\infty]{\mathbb{P}}0.   
  \end{equation}
  
   \par Next, we consider the asymptotic properties of the denominator of equation (\ref{6.8}). We can similarly expand the denominator using a Taylor series around $H$, resulting in
      \begin{align}
      	 \Phi^{{n-1}^{*} }_{1}(u) \Gamma^{-1}_{(n-1)T}(\hat H_{n}) \Phi^{n-1}_{1}(u)
      	&=
      	  \Phi^{{n-1}^{*} }_{1}(u) \Gamma^{-1}_{(n-1)T}(H) \Phi^{n-1}_{1}(u)   \\
      	& \quad+
      	\Phi^{{n-1}^{*} }_{1}(u) A^{(1)}_{nT}(H)\Phi^{n-1}_{1}(u)(\hat H_{n} - H) \nonumber \\
      	& \quad+
      	\frac{1}{2} \Phi^{{n-1}^{*} }_{1}(u) A^{(2)}_{nT}(H)\Phi^{n-1}_{1}(u)(\hat H_{n} - H)^{2},\nonumber  
      \end{align}  
 similarly, this part of the proof aligns with Lemma 1 \citep{esstafa2019long} and satisfy
      \begin{equation}
      	\frac{1}{n}\Phi^{{n-1}^{*} }_{1}(u) A^{(1)}_{nT}(H) \Phi^{n-1}_{1}(u) \xrightarrow[n\rightarrow\infty]{\mathbb{P}} k^{(3)}_{H,\phi(u)},\label{6.14}
      \end{equation}
      \begin{equation}
      	n^{-\frac{3}{2}}\Phi^{{n-1}^{*} }_{1}(u) A^{(2)}_{nT}( H )\Phi^{n-1}_{1}(u) = O_{\mathbb{P}}(1) ,\label{6.15}
      \end{equation}
      \begin{equation}
      	\frac{1}{n} \Phi^{{n-1}^{*} }_{1}(u)\Gamma^{-1}_{(n-1)T}(H) \Phi^{n-1}_{1}(u) \xrightarrow[n\rightarrow\infty]{\mathbb{P}} \frac{1}{1-\phi^{2}(u)}, \label{6.16}
      \end{equation}  
   where $k^{(3)}_{H,\phi(u)}$is a constant. The denominator in equation (\ref{6.8}) converges in probability as follows
   \begin{equation}
  \frac{1}{n}\Phi^{{n-1}^{*} }_{1}(u) \Gamma^{-1}_{(n-1)T}(\hat H_{n}) \Phi^{n-1}_{1}(u)\xrightarrow[n\rightarrow\infty]{\mathbb{P}} \frac{1}{1-\phi^{2}(u)}.
  \end{equation}
Combining the above equations, we find that when the numerator of equation (\ref{6.8}) is multiplied by $\frac{1}{n}$, it approaches $0$, while the denominator, also multiplied by $\frac{1}{n}$, converges to a constant. Furthermore, since convergence in probability implies convergence in distribution, we conclude that
\begin{equation}
\hat{\phi}_{n}(u) \xrightarrow[n\rightarrow\infty]{{\mathbb{P}}} \phi(u).
 \end{equation}

This establishes a clear relationship between the asymptotic behavior of the numerator and denominator, leading to the convergence of the estimated function.

  \subsection{Proof of Theorem 2.5}
  
 According to Theorem 2 \citep{hurvich1998mean}, without loss of generality, we can assume $m =[n^{\delta}]$ for some $\frac{1}{2}<{\delta}<  \frac{2}{3} $. denotes convergence in distribution. We thus have
\begin{equation}
\sqrt{m} (\hat{H}_{n} - H) \xrightarrow[n\rightarrow\infty]{\mathcal{L}} \mathcal{N}(0,  {U_{K,\delta}}),
\end{equation}
    where $ {U_{K,\delta}}$ is a constant related to $K$ and $\delta$. Building on the results from equation ($\ref{6.8}$), we establish that
    \begin{equation}
      	\sqrt{m}(\hat \phi_{n}(u) - \phi (u)) = \sqrt{m} \frac{\tilde \Phi^{n^{*}}_{2}(u) \Gamma^{-1}_{(n-1)T}( H) \Phi^{n-1}_{1}(u)}{\Phi^{{n-1}^{*}}_{1}(u)\Gamma^{-1}_{(n-1)T}(\hat H_{n}) \Phi^{n-1}_{1}(u)}  +R^{(1)}_{n},\label{6.34}
    \end{equation}
       according to the proof of consistency and some results on \cite{esstafa2019long},the denominator of the first term on the right side of the above equation satisfies
      \begin{equation}
      	\frac{1}{n}   \Phi^{{n-1}^{*}}_{1}(u) \Gamma^{-1}_{(n-1)T}(\hat H_{n}) \Phi^{n-1}_{1}(u) \xrightarrow[n\rightarrow\infty]{{\mathbb{P}}} \frac{1}{1-\phi^{2}(u)},  
      \end{equation}
      the nominator converge to a normal distribution 
      \begin{equation}
        \frac{1}{\sqrt{n}}\tilde \Phi^{n^{*}}_{2}(u) \Gamma^{-1}_{(n-1)T}(H) \Phi^{n-1}_{1}(u)\xrightarrow[n\rightarrow\infty]{{\mathbb{P}}} \mathcal{N}(0,\frac{1}{1-\phi^{2}(u)}),
     \end{equation}
        when $n \rightarrow \infty $,  the reminder $R^{(1)}_{n}$ converges to 0.\\
	   Thus, we can rewrite equation (\ref{6.34}) as follows
	\begin{eqnarray}
		\sqrt{m}(\hat \phi_{n}(u)-\phi(u))
		& = & 
		\sqrt{m}\frac{\frac{1}{n}\tilde \Phi^{n^{*}}_{2}(u) \Gamma^{-1}_{(n-1)T}(H) \Phi^{n-1}_{1}}{\frac{1}{n}\Phi^{{n-1}^{*}}_{1}(u)\Gamma^{-1}_{(n-1)T}({\hat H}_{n}) \Phi^{n-1}_{1}(u)}+R^{(1)}_{n}   \\
		& = & 
		\frac{\sqrt{m}}{\sqrt{n}}(1-\phi^{2}(u))\frac{1}{\sqrt{n}}\tilde \Phi^{n^{*}}_{2}(u) \Gamma^{-1}_{(n-1)T}(H) \Phi^{n-1}_{1}(u)+R^{(1)}_{n}
		,\nonumber
	\end{eqnarray}
      by slutsky theorem, we can conclude that
      $\sqrt{m}(\hat \phi_{n}(u) - \phi (u))$ converges to a  normal distribution.
   \par Lastly, we aim to present these results in the form of a joint normal distribution. Drawing on the findings from \cite{hariz2024fast}, the asymptotic distribution of $\sqrt{n}(\hat{\phi}_{n}(u)-\phi(u))$ can be expressed as a constant multiple of the asymptotic distribution of $(\hat{H}_{n}-H)$. Moreover, according to the Cramer-Wold theorem, the asymptotic distribution of $\sum_{u = 1}^{T}(\hat{\phi}_{n}(u)-\phi(u))$ still adheres to an asymptotic normal distribution.
Thus, the vector
	  $$(({\hat{\phi}}_{n}(1)-\phi(1)),({\hat{\phi}}_{n}(2)-\phi(2)),...({\hat{\phi}}_{n}(T)-\phi(T)),...\hat H_{n}-H),$$ 
     converges to a Gaussian vector, tending towards a joint normal distribution. The covariance matrix of this vector is
      \begin{equation}
	   \tilde \Sigma_{\theta} = U_{\theta}U^{*}_{\theta},
      \end{equation}
	  where
  \begin{equation}	
      U_{\theta} = \left(
	 \begin{array}{c}
		1\\
		\frac{{C^{(1)}_{H,\phi(1)}}}{({1-\phi^{2}(1)})}\\
		\vdots\\
		\frac{{C^{(1)}_{H,\phi(T)}}}{({1-\phi^{2}(T)})}
	\end{array}
	\right),
 \end{equation}
        ${C^{(1)}_{H,\phi(u)}}$ is the constants related to $\phi(u)$ and $H$.
  
  \subsection{Proof of Theorem 3.1} 
   To prove the theorem 3.1, we need to establish the following three lemmas and verify whether$ g_{H,\underline{\phi}(\lambda)}$ is regular. The regularity conditions of $ g_{H,\underline{\phi}(\lambda)}$ will be demonstrated in the auxiliary results.
     \begin{lemma}
  Let $\theta_{0} \in \Theta$, $\delta>0$, such that for any $\theta \in B(\theta_{0},\delta)$, it holds that  
    \begin{equation}
    ||\mathcal{I}(\theta)-\mathcal{I}(\theta_{0})||\leq K||\theta-\theta_{0}||,
    \end{equation}
where $K$ is some constant.
	 \end{lemma}
       \begin{proof}
       Without loss of generality, let $B(\theta_{0},\delta)$ be a convex set in $\mathbb{R}^{3}$. For ease of notation, $g_{H,\underline {\phi}}(\lambda)$ can be denoted as $g_{\theta}(\lambda)$. According to the relevant conclusions in \cite{cohen:hal-00638121} and the discussion of regularity conditions for $g_{H,\underline {\phi}}(\lambda)$, it is known that for any $k,j\in \left\{1,2,...,d\right\}$ that the following inequality holds
        \begin{equation}
           \left|\frac{1}{4\pi}(\int_{-\pi}^{\pi}\frac{\partial log \,   g_{\theta}(\lambda)}{\partial \theta_{k}}\frac{\partial log \,   g_{\theta}(\lambda)}{\partial \theta_{j}}d\lambda) -\frac{1}{4\pi}(\int_{-\pi}^{\pi}\frac{\partial log \,   g_{\theta_{0}}(\lambda)}{\partial \theta_{0,k}}\frac{\partial log \,   g_{\theta_{0}}(\lambda)}{\partial \theta_{0,j}}d\lambda)\right|
           \leq
           K||\theta-\theta_{0}||, 
       \end{equation}
       $K$ is defined as 
      \begin{equation}
          K = \sup_{\theta \in B(\theta_{0},\delta)}\left\lVert(\frac{\partial}{\partial \theta_{i}}(\int^{\pi}_{-\pi}\frac{\partial log \,   g_{\theta}(\lambda)}{\partial \theta_{k}}\frac{\partial log \,   g_{\theta}(\lambda)}{\partial \theta_{j}}d\lambda))_{1 \leq i \leq d}\right\rVert, 
      \end{equation}
         which is related to $k$ and $j$.
      Furthermore, since the conditions (A1) and (A2) \citep{cohen:hal-00638121} hold, it follows that $K< \infty$, hence the lemma holds.
     \end{proof}
     
     \begin{lemma}
     	For any $\theta \in \Theta^{l}$, it follows from the distribution of the parameter $\theta$ that
\begin{equation}
     \frac{\Delta l_{n}(\theta)}{\sqrt{n}}+\sqrt{n}\mathcal{I}(\theta) = O_{\mathbb{P}}(1).
\end{equation}
     \end{lemma}
\begin{proof}
The Lemma 3.6 \citep{cohen:hal-00638121} implies that, from the distribution of $\theta$, we have
 \begin{equation}
   E\left(\frac{\Delta l_{n}(\theta)}{n}\right)\rightarrow -\mathcal{I}(\theta).   
  \end{equation}
To determine the convergence rate of the above expression, Lemma 3 and Lemma 4 \citep{lieberman2012asymptotic} yield the following conclusion
  \begin{equation}
     E\left(\frac{\Delta l_{n}(\theta)}{n}\right)+\mathcal{I}(\theta) = O(n^{-1+\delta}),
  \end{equation}
where $\delta$ is a positive real number. Therefore,
  \begin{equation}
    E\left(\frac{\Delta l_{n}(\theta)}{\sqrt{n}}\right)+\sqrt{n}\mathcal{I}(\theta) = O(n^{-\frac{1}{2}+\delta}).
 \end{equation}
Furthermore, by utilizing Lemma 3.6 \citep{cohen:hal-00638121} once again, we obtain
\begin{equation}
    Var\left(\frac{\Delta l_{n}(\theta)}{\sqrt{n}}\right)=O(1).
\end{equation}
   Thus, the proof is concluded.
\end{proof}
     \begin{lemma}
     Let $\left\{ {\overline{\theta}}_{n}\right\}_{n}$ be a stochastic sequence satisfying ${\overline{\theta}}_{n} - \theta = o_{\mathbb{P}}(1)$. Then, according to the distribution of parameter $\theta$, for any $k>0$, it holds that
\begin{equation}
 \frac{\Delta l_{n}({\overline{\theta}}_{n})}{n} - \frac{\Delta l_{n}(\theta)}{n} = O_{\mathbb{P}}(n^{k}({\overline{\theta}}_{n}-\theta)).
\end{equation}
     \end{lemma}
      \begin{proof}
      Let $C_{k,\theta}$ be a compact convex set depending on $k$ and $\theta$, and ${\overline{\theta}}_{n} \in C_{k,\theta}$. According to the proof of Lemma 3.7 \citep{cohen:hal-00638121}, we have
 \begin{equation}
      \sup_{{\overline{\theta}}_{n} \in C_{k,\theta}}\left|\frac{\partial^{3}}{\partial^{i_{1}}{\theta_{1}}\partial^{i_{2}}{\theta_{2}}...\partial^{i_{d}}{\theta_{d}}}\frac{l_{n}({{\overline{\theta}}_{n}})}{n^{1+k}}\right| = O_{\mathbb{P}}(1),
\end{equation}
where $(i_{1},i_{2},...,i_{d}) \in \left\{0,1,2,3\right\}^{d}$, satisfying $i_{1}+i_{2}...+i_{d} =3$. 
In conclusion, for a finite positive random variable $K$, we have
\begin{equation}
     P\left(\left\lVert \frac{\Delta l_{n}({\overline{\theta}}_{n})}{n} - \frac{\Delta l_{n}(\theta)}{n}\right\rVert \leq Kn^{k}(||{\overline{\theta}}_{n}-\theta||\right) \geq P({\overline{\theta}}_{n} \in C_{k,\theta}),
\end{equation}
which implies $\frac{\Delta l_{n}({\overline{\theta}}_{n})}{n} - \frac{\Delta l_{n}(\theta)}{n} = O_{\mathbb{P}}(n^{k}({\overline{\theta}}_{n}-\theta))$ holds.
  \end{proof}
      According to the hypothesis of this theorem, we can deduce
	 \begin{equation}
	 	\sqrt{n}(\tilde \theta_{n}-\theta ) = \sqrt{n}(\hat \theta_{n}-\theta)+\mathcal{I}^{-1}(\hat \theta_{n})\frac{\nabla l_{n}(\hat \theta_{n})}{\sqrt{n}},  \label{6.18}
	 \end{equation}
applying mean-value theorem to $\nabla l_{n}(\theta)$, we have
	 \begin{equation}
	 	\nabla l_{n}(\hat \theta_{n}) = \nabla l_{n}(\theta)+(\hat \theta_{n}-\theta)\int^{1}_{0}\bigtriangleup l_{n}(\theta+v(\hat \theta_{n}-\theta))dv,\label{6.19}
	 \end{equation}
	   Substituting equation ($\ref{6.19}$) to equation ($\ref{6.18}$), we produce 
	 \begin{equation}
	 	\sqrt{n}(\tilde \theta_{n}-\theta) = \sqrt{n}(\hat \theta_{n}-\theta)\mathcal{I}^{-1}(\hat \theta_{n})(\mathcal{I}(\hat \theta_{n})+\frac{\int^{1}_{0}\bigtriangleup l_{n}(\theta(v))dv}{n})+\mathcal{I}^{-1}(\hat \theta_{n})\frac{\nabla l_{n}(\theta)}{\sqrt{n}}, \label{6.20}
	 \end{equation}
	   where $\theta(v)$ = $\theta + v(\hat \theta_{n}-\theta)$, $v < 1$.
	 Next, we will discuss  the consistency and asymptotic normality of one-step estimator.\\
	  (1) Consistency of $\tilde{\theta}_{n}$
  \par Observing equation ($\ref{6.20}$) ,  the first and second terms on the right-hand side can be expressed as  
     \begin{eqnarray}
         A_{n} 
         &=& \sqrt{n}(\hat \theta_{n}-\theta)\mathcal{I}^{-1}(\hat \theta_{n})(\mathcal{I}(\hat \theta_{n})+\frac{\int^{1}_{0}\bigtriangleup l_{n}(\theta(v))dv}{n})     \\ 
         &=& \sqrt{n^{\delta}}(\hat \theta_{n}-\theta)\mathcal{I}^{-1}(\hat \theta_{n})\sqrt{n^{1-\delta}}(\mathcal{I}(\hat \theta_{n})+\frac{\int^{1}_{0}\bigtriangleup l_{n}(\theta(v))dv}{n}) \nonumber
     \end{eqnarray}
       and
    \begin{equation}
      B_{n}=\mathcal{I}^{-1}(\hat \theta_{n})\frac{\nabla l_{n}(\theta)}{\sqrt{n}}=\mathcal{I}^{-1}(\theta)\frac{\nabla l_{n}(\theta)}{\sqrt{n}}+(\mathcal{I}^{-1}(\hat \theta_{n})-\mathcal{I}^{-1}(\theta))\frac{\nabla l_{n}(\theta)}{\sqrt{n}}.
    \end{equation}
      Fristly, we analyze the properties of $A_{n}$ and derive the following equation
   \begin{equation} \label{t5}
     \begin{aligned}
    \mathcal{I}(\hat \theta_{n})+\frac{\int^{1}_{0}\Delta l_{n}(\theta(v))dv}{n}
    &= (\mathcal{I}(\hat \theta_{n})-\mathcal{I}(\theta_{n})) \\
    &\quad +(\mathcal{I}(\theta_{n})+\frac{\Delta l_{n}(\theta)}{n})+\frac{1}{\sqrt{n}}\int^{1}_{0}(\frac{\Delta l_{n}(\theta(v)))}{\sqrt{n}}-\frac{\Delta l_{n}(\theta)}{\sqrt{n}})dv,
   \end{aligned}
   \end{equation}
     based on equation ($\ref{t5}$) and lemmas 6.2, 6.3, and 6.4. The convergence order of $\frac{A_{n}}{\sqrt{n}}$ is
   \begin{equation}   
      \frac{A_{n}}{\sqrt{n}} = n^{-\frac{\delta}{2}}(O_{\mathbb{P}}(n^{-\frac{\delta}{2}})+O_{\mathbb{P}}(n^{-\frac{1}{2}})+O_{\mathbb{P}}(n^{k-\frac{\delta}{2}})),
 \end{equation}
        when $k-\delta < 0$, we have $\frac{A_{n}}{\sqrt{n}}\xrightarrow[n\rightarrow\infty]{\mathbb{P}}0$.\\
Secondly, we consider the property of $B_{n}$ and it has the form of
   \begin{equation}
   B_{n}=\mathcal{I}^{-1}(\hat \theta_{n})\frac{\nabla l_{n}(\theta)}{\sqrt{n}}=\mathcal{I}^{-1}(\theta)\frac{\nabla l_{n}(\theta)}{\sqrt{n}}+(\mathcal{I}^{-1}(\hat \theta_{n})-\mathcal{I}^{-1}(\theta))\frac{\nabla l_{n}(\theta)}{\sqrt{n}}, 
   \end{equation}
   according to \cite{hariz2024fast} and theorem 1 in \cite{lieberman2012asymptotic}, we have
   \begin{equation}
    \frac{\nabla  l_{n}({\theta})}{\sqrt{n}} \xrightarrow[n\rightarrow\infty]{\mathbb{P}} 0.
   \end{equation}
   When $\mathcal{I}n(\cdot)$ is a non-degenerate continuous function, as indicated by the above equation, it can be observed that both the first and second terms of $B_{n}$ tend to 0. Consequently, $\frac{B_{n}}{\sqrt{n}}$ converges in probability to 0, and naturally, it also converges in distribution to 0.
   \par Combining the above results, we can conclude the consistency of $\tilde \theta_{n}$. \\
(2) Asymptotic normality of $\tilde{\theta}_{n}$ 
    \par According to the results of \cite{hariz2024fast}, the equation
	\begin{equation}
        \mathcal{I}^{-1}(\theta)\frac{\nabla  l_{n}(\theta)}{\sqrt{n}}+(\mathcal{I}^{-1}(\hat \theta_{n})-\mathcal{I}^{-1}(\theta))\frac{\nabla  l_{n}(\theta)}{\sqrt{n}} \label{6.60}
    \end{equation}
    converges in probability to a bounded limit as $n \rightarrow \infty$. Simultaneously, the second term on the right-hand side of equation ($\ref{6.60}$) converges to 0. By applying the Slutsky theorem, we can verify the asymptotic normality of $\tilde \theta_{n}$.

  \section{Auxiliary results}
     \begin{lemma}
      Under the hypothesis on the parametric space have the following results\\
   (1) For any $H \in [0,1]$ and $j \in \left\{0,1,2,3 \right\}$, $\frac{\partial}{\partial \lambda } \frac{\partial^{j}}{\partial^{j}H}g_{H,\underline {\phi}}(\lambda)$.\\
   (2) For any $j \in \left\{0,1,2,3 \right\} $ the functions $\frac{\partial^{j}}{\partial^{j}H}g_{H,\underline {\phi}}(\lambda)$ are symmetric with respect to  $\lambda$.\\
   (3) For any $\delta > 0 $ and all $(H,\lambda) \in [0,1] \times  [-\pi,\pi]  \backslash\left\{0 \right\}$
    \par a.$	C_{1,\delta} |\lambda|^{1-2H+\delta }\leq g_{H,\underline {\phi}}(\lambda) \leq C_{2,\delta}|\lambda|^{1-2H-\delta }.$ 
    \par b.$|\frac{\partial}{\partial \lambda}g_{H,\underline {\phi}}(\lambda)| \leq C_{3,\delta}|\lambda|^{-2H-\delta}.$
    \par c.For any $j \in \left\{0,1,2,3 \right\} $, $|\frac{\partial^{j}}{\partial^{j} H}g_{H,\underline {\phi}}(\lambda)| \leq C_{4,\delta}|\lambda|^{-2H-\delta},$\\
 where $C_{i,\delta}$ are constants for $i = 1,2,3,4.$
     \end{lemma}
    \begin{proof}
     We start from Assertion 3a, which states that
     \begin{equation}
         g_{H,\underline {\phi}}(\lambda)=  C_{H}|\sum^{T-1}_{p=0}\frac{e^{-ip\lambda}}{1-\phi(T)e^{-i\lambda T}} |^{2}(1-cos(\lambda))\sum\limits_{j \in Z}\frac{1}{|\lambda+2j\pi|^{2H+1}},
    \end{equation}
     where $C_H=\frac{1}{2\pi}\Gamma(2H+1)sin(\pi H)$ and $\Gamma(\cdot)$ denote the Gamma function. 
     According to Lemma 5.4 in \cite{hariz2024fast}, we have
     \begin{equation}
     K_{1,\delta}|\lambda|^{1-2H+\delta } \leq C_{H}(1-cos(\lambda))\sum\limits_{j \in Z}\frac{1}{|\lambda+2j\pi|^{2H+1}} \leq K_{2,\delta}|\lambda|^{1-2H-\delta }  
     \end{equation}
     and 
     \begin{equation}
        \frac{\sqrt{2}T}{1-\phi^{2}_{max}(u)}\leq |\sum^{T-1}_{p=0}\frac{e^{-ip\lambda}}{1-\phi(T)e^{-i\lambda T}} |^{2} \leq \frac{\sqrt{2}T}{1-\phi^{2}_{min}(u)},
     \end{equation}
     where $\phi^{2}_{max}(u)$ and $\phi^{2}_{min}(u)$ are the  maximum and minimum values of $\phi(u)$, respectively. Thus, Assertion 3a has been proved and assertion 3b follows straightforwardly from Assertion 3a. 
\par Next, we discuss Assertion 3c, which can be obtained directly from  Lemma 5.4 in \cite{hariz2024fast}. The partial derivative of  $g_{H,\underline {\phi}}(\lambda)$ does not depend on $|\sum^{T-1}_{p=0}\frac{e^{-ip\lambda}}{1-\phi(T)e^{-i\lambda T}} |^{2}$, and the modulus is bounded.\\    
 \end{proof}

\end{document}